\documentclass[journal]{IEEEtran}
\usepackage{amsmath,amsfonts}
\usepackage{algorithmic}
\usepackage{algorithm}
\usepackage{array}
\usepackage[caption=false,font=scriptsize,labelfont=sf,textfont=sf]{subfig}
\usepackage{hyperref}
\usepackage{amssymb,amsfonts,bm,amsthm,amscd}
\usepackage{textcomp}
\usepackage{stfloats}
\usepackage{url}
\usepackage{verbatim}
\usepackage{graphicx}
\usepackage{float}
\usepackage{multirow}
\usepackage{caption}
\usepackage{cite}
\usepackage{color,xcolor}
\usepackage{makecell}
\graphicspath{{figures/}}
\newtheorem{theorem}{Theorem}[section]

\newtheorem{proposition}{Proposition}[section]
\newtheorem{definition}{Definition}[section]
\newtheorem{remark}{Remark}[section]

\begin{document}

\title{Tensor-based Model Reduction and Identification for Generalized Memory Polynomial}

\author{Yuchao Wang\thanks{This paper was produced by the IEEE Publication Technology Group. They are in Piscataway, NJ.}
\thanks{Yuchao Wang is with the School of
	Mathematical Sciences, Fudan University, Shanghai, 200433, P. R. of
	China (e-mail: 21110180041@m.fudan.edu.cn). This author is supported by the Science and Technology
	Commission of Shanghai Municipality under grant 23JC1400501 and  the National Natural Science Foundation of China under grant 12271108. } and Yimin Wei \thanks{Yimin Wei is with the School of Mathematical Sciences and  Key Laboratory of Mathematics for Nonlinear
	Sciences,  Fudan University, Shanghai, 200433, P. R. of China (e-mail: ymwei@fudan.edu.cn).This author is supported by the National Natural Science Foundation of China under grant 12271108 and the Science and Technology
	Commission of Shanghai Municipality under grant 23JC1400501.}}

\markboth{IEEE TRANSACTIONS ON AUTOMATION SCIENCE AND ENGINEERING,~Vol.~xx, No.~xx, 2025}%
{Shell \MakeLowercase{\textit{et al.}}: A Sample Article Using IEEEtran.cls for IEEE Journals}


\maketitle
 
\begin{abstract}
Power amplifiers (PAs) are essential components in wireless communication systems, and the design of their behavioral models has been an important research topic for many years. The widely used generalized memory polynomial (GMP) model suffers from rapid growth in the number of parameters with increasing memory depths and nonlinearity order, which leads to a significant increase in model complexity and the risk of overfitting. In this study, we introduce tensor networks to compress the unknown coefficient tensor of the GMP model, resulting in three novel tensor-based GMP models. These models can achieve comparable performance to the GMP model, but with far fewer parameters and lower complexity. For the identification of these models, we derive the alternating least-squares (ALS) method to ensure the rapid updates and convergence of model parameters in an iterative manner. In addition, we notice that the horizontal slices of the third-order data tensor constructed from the input signals are Vandermonde matrices, which have a numerically low-rank structure. Hence, we further propose the RP-ALS algorithm, which first performs a truncated higher-order singular value decomposition on the data tensor to generate random projections, then conducts the ALS algorithm for the identification of projected models with downscaled dimensions, thus reducing the computational effort of the iterative process. The experimental results show that the proposed models outperform the full GMP model and sparse GMP model via LASSO regression in terms of the reduction in the number of parameters and running complexity.
\end{abstract}

\def\abstractname{Note to Practitioners}
\begin{abstract}
	This paper is motivated by the demand for low-complexity and accurate PA behavioral models, offering a viable option for modeling wideband PAs in modern wireless communication systems. A good model should provide satisfactory performance with as low complexity as possible and allow for quick identification of model parameters. Based on the advantages of tensor networks, this paper proposes three tensor-based GMP models that can utilize the potential low-rank tensor structure of the GMP model, and derives the alternating least-squares method to rapidly identify the model parameters. These models can efficiently reduce the number of parameters and running complexity while maintaining comparable accuracy to the GMP model.
\end{abstract}

\begin{IEEEkeywords}
Power amplifier, generalized memory polynomial, tensor network, randomized projection.
\end{IEEEkeywords}

\section{Introduction}
\IEEEPARstart{I}{n} wireless communication systems, RF PAs are crucial components for boosting the power of transmitted signals, but they also cause nonlinear distortions when operating near saturation. Digital predistorters \cite{ghannouchi2009behavioral, kim2001digital, ding2004robust}, which have gained widespread popularity, are employed to compensate for the distortions provided that the nonlinear behaviors of the PAs can be accurately modeled and estimated. Theoretically, Volterra series \cite[Chapter 6]{2014Behavioral} can approximate to arbitrary precision any nonlinear system with fading memory, but is not ideal for characterizing PAs due to their tremendous parameters, which greatly increase the computational complexity and data resources. Many simplified models, e.g., memory polynomial \cite{kim2001digital, ding2004robust} and generalized memory polynomial (GMP) model \cite{Morgan20063852} have emerged in the past two decades. However, the demand for low-complexity and high-precision PA models has been increasing with the development of communication technology. Tensor networks and compressed sensing are two effective mathematical methods for model reduction that have been implemented to improve some existing PA models.

Tensor networks, including
CANDECOMP/PARAFAC (CP) model \cite{hitchcock1927expression}, Tucker decomposition \cite{tucker1966some} and the tensor-train decomposition \cite{oseledets2011tensor}, have been extensively studied and successfully applied to machine learning \cite{luo2021adjusting, wang2021augmented, shen2022robust,lee2021qttnet}, signal processing \cite{sidiropoulos2017tensor, chang2022general, cichocki2015tensor} and many other fields. It is worth noting that the experts formulated the Volterra series model in tensorial format using tensor-vector mode product, and leveraged tensor networks to represent the coefficient kernels, leading to the tensor-based Volterra models \cite{favier2009parametric, Favier201230, Batselier201726, Batselier2016, chen2017tensor}. To some extent, these models mitigate the fatal problems of Volterra series that the number of parameters and model complexity exponentially explode with memory depth and nonlinearity order, but their identification cost and running complexity are still too high in some scenarios.

To reduce the number of parameters and model complexity, compressed sensing theory has been applied to the sparse recovery of some existing PA models as well \cite{becerra2020sparse, abdelhafiz2014digital, reina2015behavioral, becerra2018doubly, li2016sparsity, wang2023pruning}. For instance, the compressed sampling matching pursuit algorithm was employed for various well-known predistorters to reduce their complexity \cite{abdelhafiz2014digital}. Under the sparsity assumption for the coefficient kernels, the greedy algorithm is adopted for the sparse recovery of the Volterra series model \cite{becerra2020sparse,becerra2018doubly,reina2015behavioral}. In \cite{wang2023pruning}, the least absolute shrinkage and selection operator (LASSO) regression is exerted for pruning redundant terms in the GMP model.

The GMP model introduced by Morgan et al. \cite{Morgan20063852} is a widely used behavioral model for PAs. It has attracted much attention from researchers and practitioners. In 2010, the paper \cite{tehrani2010comparative} presented a comparative analysis of the state-of-the-art PA behavioral models, and emphasized that the GMP model has the best trade-off for accuracy versus complexity. The determination of the optimal nonlinearity order and memory depths for GMP model is also studied by employing various algorithms, e.g., hill-climbing heuristic \cite{wang2018novel}, genetic algorithm \cite{mondal2013genetic}, particle swarm optimization \cite{abdelhafiz2018generalized} and artificial bee colony optimization \cite{deepak2020identification}. From the perspective of reducing the model complexity, some modified models are proposed by deleting some cross terms or splitting the coefficients based on GMP model \cite{chen2021modified, liu2013robust}, and the LASSO regression is applied to extract the sparsity of GMP model \cite{wang2023pruning}, but the computational cost of model identification increases in the meantime. In this paper, to exploit the potential tensor structure of the GMP model, we introduce three tensor-based GMP models that can achieve comparable accuracy but with far fewer model parameters and lower complexity, and derive fast algorithms to identify these models. Our work can be mainly divided into the following parts:
\begin{enumerate}
	\item{We propose new descriptions of GMP model using tensor networks, termed GMP-CP, GMP-TT and GMP-Tucker model respectively. Through the complexity analysis and experimental validation, these models can efficiently reduce the number of parameters and model complexity while maintaining the comparable accuracy.}
	\item{The identifications of these tensor-based GMP models become multilinear least-squares problems, we derive the alternating least-squares (ALS) algorithm to solve them in an iterative manner, which ensures rapid updates and convergence of model parameters.}
	\item{We use the randomized projections, computed from the truncated higher-order singular value decomposition (HOSVD) on the data tensor, to reduce the model dimensions. Accordingly, the computational cost of iterative process for model identification is further reduced.}
\end{enumerate}

This paper is organized as follows. Section \ref{sec. 2} introduces the related works, including the basic tensor theory and GMP model. In Section \ref{sec. 3}, we propose the GMP-CP, GMP-TT, and GMP-Tucker models, and deduce the ALS algorithm for their identifications; the model complexity analysis is also presented. In Section \ref{sec. 4}, we demonstrate how to reduce model dimensions by random projections if the fast truncated HOSVD on the data tensor is performed, based on which, an algorithm called RP-ALS is proposed to further accelerate the model identification process. In Section \ref{sec. 5}, we present experimental results on simulated PA datasets, demonstrating the superior performance of our models and algorithms. Finally, we summarize our findings and conclusions in Section \ref{sec. 6}.

\section{Notations and Preliminaries}
\label{sec. 2}
Tensors are multidimensional arrays. Specifically, a vector is a first-order tensor and a matrix is a second-order one. The notation $\mathbb{F}^{I_1 \times I_2 \times \cdots \times I_d}$ denotes the set of $d$th-order tensors of size $I_1 \times I_2 \times \cdots \times I_d$ over number field $\mathbb{F}$, which can be the real number field $\mathbb{R}$ or the complex number field $\mathbb{C}$. The dimensions $I_k \in \mathbb{N}$ for $k=1,2,\ldots, d$ are called tensor modes. Unless specifically stated in the sequel, we always denote by calligraphic capital letters $\mathcal{X}, \mathcal{S}, \ldots$ the high-order tensors, while boldfaced capital letters $\mathbf{A}, \mathbf{B}, \ldots$ for matrices, boldfaced lowercase letters $\mathbf{a}, \mathbf{b}, \ldots$ for vectors, and lowercase letters for scalars. The entries of tensor $\mathcal{X} \in \mathbb{C}^{I_1 \times I_2 \times \cdots \times I_d}$ are accessed by $\mathcal{X}_{i_1 i_2 \cdots i_d} \in \mathbb{C}$, and the elements of matrix $\mathbf{A}$ are denoted by $a_{i j}$ or $(\mathbf{A})_{i j}$, analogously for vectors. The mode-$k$ unfolding matrix (see Definition \ref{Def: modekUnfold}) of a tensor $\mathcal{X}$ is denoted by $\mathbf{X}_{(k)}$. The conjugate transpose of a matrix $\mathbf{A}$ is denoted by $\mathbf{A}^*$. The colon notation `$:$' is used to indicate the free indices (with the same meaning as in MATLAB), for example, $\mathcal{X}_{i_1 \cdots i_{k-1}: i_k \cdots i_d}$ is a vector obtained by fixing all indices but the $k$th one, which is called the mode-$k$ fiber of $\mathcal{X}$. Table~\ref{table: Nomenclature} shows the frequently used common notations in this article.

\begin{table}[htbp!]
\centering
\caption{Nomenclature}
\label{table: Nomenclature}
\begin{tabular}{|c|l|}
	\hline $x(t),y(t)$ & 
	The input, output complex signals of PA models \\
	\hline $\mathcal{S}$ & 
		\makecell[l]{The third-order tensor formed by the  coefficients \\ of GMP model} \\
	\hline $\mathcal{X}$ & 
		\makecell[l]{The fourth-order tensor computed from the basis \\ functions of GMP model} \\
	\hline $\mathcal{M}$ & \makecell[l]{The third-order tensor computed from the basis\\ terms of tensor-based GMP models} \\
	\hline $\widetilde{\mathcal{M}}$ & The core tensor of truncated HOSVD of $\mathcal{M}$ \\
	\hline $\widehat{\mathcal{M}}$ & The approximate tensor of $\mathcal{M}$ from truncated HOSVD  \\
	\hline $\mathbf{H}$ & The matrix formed by the input complex signals \\
	\hline $\mathbf{A}, \mathbf{B}, \mathbf{C}, \mathcal{B}$ & The parameters of tensor-based GMP models \\
	\hline $\mathbf{I}$ & The identity matrix \\
	\hline $\mathbf{y}$ & The vector formed by output signals in a time interval \\
	\hline$M_1, M_2$ & The memory depths of GMP model \\
	\hline $P$ & The nonlinearity order of GMP model \\
	\hline$\|\cdot\|_F$ & \makecell[l]{The Frobenious norm of any order tensors, and is \\ equivalent to the $\ell_2$ norm for vectors} \\
	\hline$\|\cdot\|_2$ & The spectral norm of matrices \\
	\hline$\|\cdot\|_1$ & \makecell[l]{The $\ell _1$ norm of vectors, i.e., the sum of the modulus of \\ all components} \\
	\hline$\|\cdot\|_{\infty}$ & \makecell[l]{The infinite norm of vectors, i.e., the maximum of the \\ modulus of its components} \\
	\hline vec &  The vectorization operation for tensors \\
	\hline $\circ, \otimes$, $\odot$ &  The outer, Kronecker, Hadamard product respectively \\
	\hline	
\end{tabular}
\end{table}

For the visualization of tensor networks, we use the following graphical representations of tensors and multiplication. As shown in Fig. \ref{Graph: 1}, a blue block attached to several lines represents a tensor, each line represents an index, and the common lines between two connected blocks represent the contraction with respect to these indices.
\begin{figure}[htbp!]
	\centering
	\subfloat{
		\includegraphics[width=0.0475\textwidth]{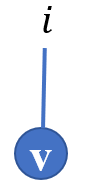}} \quad 
	\subfloat{
		\includegraphics[width=0.08\textwidth]{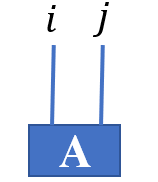}} \quad
	\subfloat{
		\includegraphics[width=0.12\textwidth]{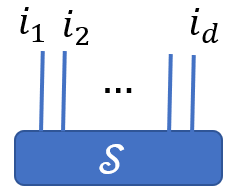}}  \quad
	\subfloat{
		\includegraphics[width=0.135\textwidth]{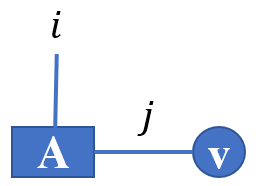}} 
	\caption{The graphical representations of a vector $\mathbf{v}$, a matrix $\mathbf{A}$ and a $d$th-order tensor $\mathcal{S}$, as well as the multiplication $\mathbf{Av}$.}
	\label{Graph: 1}
\end{figure}

\subsection{Basic tensor operations}
We now introduce some classical tensor operations from the literature \cite{kolda2009tensor, oseledets2011tensor}. Tensor contractions, also known as the Einstein summation, arise naturally in various fields such as differential geometry, quantum physics, and high-dimensional data analysis. For two tensors with some equal dimensions, the contraction can be performed on these modes; for example, if $\mathcal{X} \in \mathbb{C}^{N \times I_1 \times I_2 \times I_3}$ and $\mathcal{S}\in \mathbb{C}^{I_1\times I_2 \times I_3}$, then their contraction on the $\{2,3,4\}$-th modes of $\mathcal{X}$ with the $\{1,2,3\}$-th modes of $\mathcal{S}$ becomes a vector $\mathcal{X}\times _{2,3,4}^{1,2,3}\mathcal{S} \in \mathbb{C}^{N}$, to be precise,
\begin{equation*}
	\left(\mathcal{X}\times _{2,3,4}^{1,2,3}\mathcal{S}\right)_n = \sum_{i_1=1}^{I_1}\sum_{i_2=1}^{I_2}\sum_{i_3=1}^{I_3}  \mathcal{X}_{ni_1i_2i_3}\mathcal{S}_{i_1i_2i_3}.
\end{equation*} 

\begin{definition}[Tensor mode-$k$ product]
	The tensor mode-$k$ product is a special case of tensor contractions. The mode-k product of a tensor $\mathcal{X}\in \mathbb{C}^{I_1\times I_2\times\cdots \times I_d}$ with a matrix $\mathbf{Q}\in \mathbb{C}^{n \times I_k}$ is a $d$th-order tensor of size $I_1\times \cdots \times I_{k-1}\times n \times I_{k+1} \times \cdots I_d$, 
	\[ \left(\mathcal{X}\times _k \mathbf{Q}\right)_{i_1\cdots i_{k-1}ji_{k+1}\cdots i_d } = \sum_{i_k = 1}^{I_k} \mathcal{X}_{i_1\cdots i_{k-1}i_k i_{k+1}\cdots i_d} q_{ji_k}, \]
	it is actually the contraction on the $k$th mode of $\mathcal{X}$ and the second mode of $\mathbf{Q}$. Note that for simplicity and consistency with the literature, we omit the superscript in the contraction
	operation $\times _k^2$ to denote the mode-$k$ product.
	
	Likewise, the mode-$k$ product of $\mathcal{X}$ with a vector $\mathbf{v} \in \mathbb{C}^{I_k}$ is a $(d-1)$th-order tensor of size $I_1\times \cdots \times I_{k-1}\times I_{k+1}\times \cdots \times I_d$, 
	\[ \left( \mathcal{X} \times_k \mathbf{v} \right)_{i_1\cdots i_{k-1}i_{k+1}\cdots i_d} = \sum_{i_k = 1}^{I_k} \mathcal{X}_{i_1\cdots i_{k-1}i_k i_{k+1}\cdots i_d} v_{i_k}. \]
\end{definition}

Tensor unfolding is a frequently used operation in tensor algebraic analysis and numerical computations \cite{kolda2009tensor, Brazell2013}.  For a given sequence $\left \lfloor I \right \rfloor  = \{ I_1, I_2,\ldots,I_d \}$ of tensor dimensions and an index sequence $\left \lfloor i \right \rfloor  = \{ i_1,i_2,\ldots,  i_d \}$, we denote by $\phi(\cdot  , \cdot) : \mathbb{Z}_+^d \times \mathbb{Z}_+^d \longrightarrow  \mathbb{Z}_+$ the bijection of the reverse lexicographic ordering 
\[ \phi(\left \lfloor i \right \rfloor, \left \lfloor I \right \rfloor )=i_1+\sum_{k=2}^{d} (i_k-1)\prod_{l=1}^{k-1}I_l. \]
In fact, tensor unfolding is the process of merging multiple indices based on the bijection $\phi$.
\begin{definition}[Tensor vectorization]
	The vectorization of a tensor $\mathcal{X}\in \mathbb{C}^{I_1\times I_2\times\cdots \times I_d}$, denoted as $\mathrm{vec}(\mathcal{X})$, is a long column vector $\mathbf{v} = \mathrm{vec}(\mathcal{X})\in \mathbb{C}^{\prod_{k=1}^{d}I_k}$ whose elements are
	$v_l = \mathcal{X}_{i_1 i_2\cdots i_d}$
	with $l = \phi(\{i_1,i_2,\ldots,i_d\} , \{ I_1,I_2,\ldots,I_d \})$.
\end{definition}
\begin{definition}[Tensor mode-$k$ unfolding] \label{Def: modekUnfold}
	For a tensor $\mathcal{X}\in \mathbb{C}^{I_1\times I_2\times \cdots \times I_d}$, its mode-$k$ unfolding matrix $\mathbf{X}_{(k)} \in \mathbb{C}^{I_k \times \prod_{j = 1, j\ne k}^{d} I_j}$ is obtained by arranging the mode-$k$ fibers $\mathcal{X}_{i_1\cdots i_{k-1}:i_{k+1}\cdots i_d}$ as column vectors, for concreteness,
	\[ (\mathbf{X}_{(k)})_{: l} = \mathcal{X}_{i_1\cdots i_{k-1}:i_{k+1}\cdots i_d} \]
	with $l = \phi (\{ i_1,\dots,i_{k-1}, i_{k+1}, \dots, i_d \}, \left\{ I_1, \dots , I_{k-1},I_{k+1}, \right. \\ \left. \dots, I_d \right\})$, that is, the indices of modes
	$\{ 1, \dots, k-1, k+1, \dots, d \}$ are merged to the column index based on the reverse lexicographic ordering $\phi$.
\end{definition}
Fig. \ref{Graph: unfolding} shows the process of unfolding a third-order tensor $\mathcal{X}$ to the mode-$1$ unfolding matrix $\mathbf{X}_{(1)}$.
\begin{figure}[!htbp]
	\includegraphics[width=0.45\textwidth]{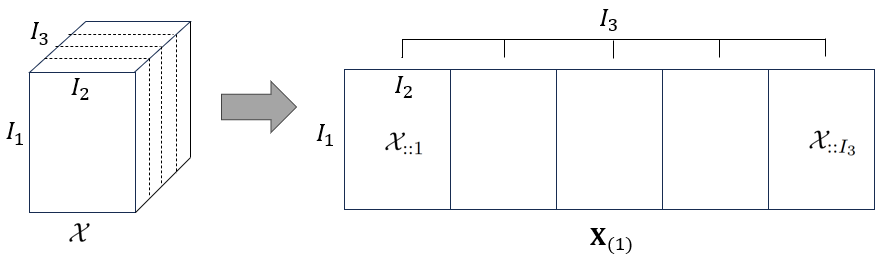}
	\centering
	\caption{The mode-$1$ unfolding process of a third-order tensor.}
	\label{Graph: unfolding}
\end{figure}

\begin{definition}[Outer product and rank-one tensor]
	The outer product of two vectors $\mathbf{a}=(a_1,a_2, \dots, a_m)^\top \in \mathbb{C}^{m}, \mathbf{b}=(b_1,b_2, \dots, b_n)^\top \in \mathbb{C}^{n}$ is a matrix $\mathbf{a}\circ \mathbf{b} \in \mathbb{C}^{m\times n}$ of rank one, and
$ \left(\mathbf{a}\circ \mathbf{b}\right)_{ij} = a_i b_j$. 

	The outer product of $d$ vectors $\mathbf{a}^{(k)} \in \mathbb{C}^{I_k}$ for $k=1,2,\dots,d$ is a $d$th-order rank-one tensor $\mathcal{X} := \mathbf{a}^{(1)} \circ \mathbf{a}^{(2)} \circ \cdots \circ \mathbf{a}^{(d)} \in \mathbb{C}^{I_1 \times I_2 \times \cdots \times I_d}$ with entries
	\[ \mathcal{X}_{i_1\cdots i_d} = \prod_{k=1}^{d} a^{(k)}_{i_k}. \]    	
\end{definition} 

\begin{definition}[Kronecker and Hadamard product]
	If $\mathbf{A} \in \mathbb{C}^{m\times n}$ and $\mathbf{B} \in \mathbb{C}^{p\times q}$, then their Kronecker product $\mathbf{A}\otimes \mathbf{B}$ is the $mp\times nq$ block matrix
	\begin{equation*}
		\mathbf{A} \otimes \mathbf{B}=\left[\begin{array}{cccc}
			a_{11} \mathbf{B} & a_{12} \mathbf{B} & \cdots & a_{1 n} \mathbf{B} \\
			a_{21} \mathbf{B} & a_{22} \mathbf{B} & \cdots & a_{2 n} \mathbf{B} \\
			\vdots &\vdots  &\ddots & \vdots \\
			a_{m 1} \mathbf{B} & 	a_{m 2} \mathbf{B} & \cdots & a_{m n} \mathbf{B}
		\end{array}\right].
	\end{equation*}
	For two matrices $\mathbf{A}$ and $\mathbf{B}$ of the same dimension $m \times n$, their Hadamard product $\mathbf{C}=\mathbf{A} \odot \mathbf{B}$ is a matrix of the same dimension with elements
	$ c_{i j}=a_{i j} b_{i j}$.
	 
	Vectors, as a special kind of matrices, are also the objects of these products.
\end{definition}

The following properties about the tensor mode-$k$ product and mode-$k$ unfolding are easily verified.
\begin{proposition}[\cite{wei2016theory}]$  $
	Let $\mathcal{X}\in \mathbb{C}^{I_1\times I_2\times\cdots \times I_d}$, $\mathbf{Q}^{(1)} \in \mathbb{C}^{n \times I_k}$, $\mathbf{Q}^{(2)} \in \mathbb{C}^{m\times n}$, and $\mathbf{v}^{(k)} \in \mathbb{C}^{I_k}$ for $k = 2,3,\ldots, d$. Then
	\begin{enumerate} \rm
		\item $\mathcal{X}\times_k \mathbf{Q}^{(1)} \times _k \mathbf{Q}^{(2)} = \mathcal{X} \times _k \left( \mathbf{Q}^{(2)}\mathbf{Q}^{(1)} \right)$.
		\item $\mathcal{X} \times _2 \mathbf{v}^{(2)} \times_3 \cdots \times_d \mathbf{v}^{(d)} = \mathbf{X}_{(1)} \left( \mathbf{v}^{(d)}\otimes \cdots \otimes \mathbf{v}^{(2)} \right)$.
	\end{enumerate}
\end{proposition}

The Frobenius norm of a tensor $\mathcal{X} \in \mathbb{C}^{I_1 \times I_2\times \cdots \times I_d}$ is
\[ 
\|\mathcal{X}\|_F=\sqrt{\sum_{i_1=1}^{I_1} \cdots \sum_{i_d=1}^{I_d}\left|\mathcal{X}_{i_1 i_2 \cdots i_d}\right|^2}, \]
which is a generalization of the Frobenius norm of matrices and vectors.

\subsection{Tensor networks}
Directly storing a higher-order tensor $\mathcal{X}^{I_1 \times I_2 \times\cdots \times I_d}$ in its original form is highly storage-intensive, since the number of elements $ {\textstyle \prod_{k=1}^{d}}I_k $ grows exponentially with the order $d$. This is known as the curse of dimensionality. Luckily, several elegant tensor networks provide low-parametric representations for tensors. 

\textbf{CP decomposition}. The canonical form of a tensor  $\mathcal{S}\in \mathbb{C}^{I_1\times I_2\times \cdots \times I_d}$ is the sum of rank-one tensors
\begin{equation*}\label{CP}
	\mathcal{S} = \sum_{r=1}^{R} \mathbf{U}_{:r}^{(1)} \circ \mathbf{U}_{:r}^{(2)} \circ\cdots \circ \mathbf{U}_{:r}^{(d)}, 
\end{equation*}
in componentwise,
\begin{equation*}
	\mathcal{S}_{i_1\cdots i_d} = \sum_{r=1}^{R} u^{(1)}_{i_1r} u^{(2)}_{i_2r}\cdots u^{(d)}_{i_kr},
\end{equation*}
where $ \mathbf{U}^{(k)} = \begin{bmatrix}
	\mathbf{U}_{:1}^{(k)}  & \mathbf{U}_{:2}^{(k)}  & \ldots & \mathbf{U}_{:R}^{(k)}
\end{bmatrix} \in \mathbb{C}^{I_k \times R}$ for $k = 1,2,\ldots,d$ are called factor matrices of $\mathcal{S}$. CP decomposition \cite{hitchcock1927expression, kolda2009tensor} devotes to find the minimum $R$ such that the above equality holds, and the number of terms $R$ is called CP-rank. Fig. \ref{Graph: CP} is the graphical illustration for
CP decomposition. In CP decomposition, we only need to record the factor matrices with $R {\textstyle \sum_{k=1}^{d}}I_k $ elements to represent a large-scale tensor, which becomes linearly dependent on the dimensions. If the considered tensor owns a potential low CP-rank structure, then it can greatly reduce the number of parameters.
\begin{figure}[!htbp]
	\centering
	\includegraphics[width=0.45\textwidth]{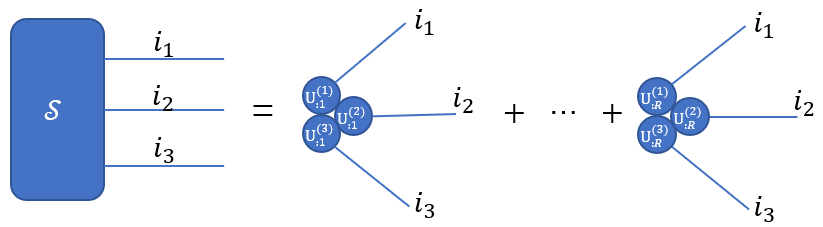}
	\caption{Graphical illustration of the CP format for a third-order tensor.}
	\label{Graph: CP}
\end{figure}

\textbf{The tensor train decomposition}. The tensor-train (TT) format \cite{oseledets2011tensor, oseledets2010tt} is another low-parametric representation for tensors, and has realized broad applications in machine learning and scientific computing, etc. In the TT decomposition, a $d$th-order tensor $\mathcal{S} \in \mathbb{C}^{I_1\times I_2\times \cdots \times I_k}$ is represented as the contractions of $d$ third-order tensors, and each entry is determined by
\begin{equation*}\label{TT}
	\mathcal{S}_{i_1\cdots i_d} = \sum_{r_0=1}^{R_0}\sum_{r_1=1}^{R_1}  \cdots \sum_{r_d=1}^{R_d}  \mathcal{G}^{(1)}_{r_0i_1r_1}\mathcal{G}^{(2)}_{r_1 i_2 r_2} \cdots \mathcal{G}^{(d)}_{r_{d-1} i_dr_d}, 
\end{equation*}
equivalently,
\begin{equation*}\label{TT-component}
	\mathcal{S} = \sum_{r_0=1}^{R_0}\sum_{r_1=1}^{R_1}  \cdots \sum_{r_d=1}^{R_d} \mathcal{G}^{(1)}_{r_0:r_1}\circ \mathcal{G}^{(2)}_{r_1 : r_2}\circ \cdots \circ \mathcal{G}^{(d)}_{r_{d-1}:r_d}, 
\end{equation*}
where $R_0=R_d=1$, these third-order tensors $\mathcal{G}^{(k)} \in \mathbb{C}^{R_{k-1}\times I_k \times R_k}$ for $k=1,2,\ldots,d$ are called TT-cores, and $(R_1, R_2,\ldots, R_{d-1})$ are TT-ranks. Fig. \ref{Graph: TT} is a graphical representation of the TT decomposition of a $d$th-order tensor, which looks like a train as its name means. In virtue of the TT decomposition, the tensor $\mathcal{S} \in \mathbb{C}^{I_1\times I_2\times \cdots \times I_k}$ can be compressed into $ {\textstyle \sum_{k=1}^{d}} R_{k-1}I_kR_k $ parameters. Lower TT-ranks imply lower memory consumption and computational cost. 
\begin{figure}[htbp!]
	\centering
	\includegraphics[width=0.49\textwidth]{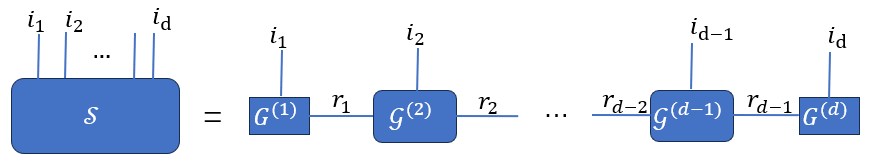}
	\caption{Graphical illustration of the TT decomposition for a $d$th-order tensor.}
	\label{Graph: TT}
\end{figure}

\textbf{Tucker decomposition}. Tucker decomposition \cite{tucker1966some} aims to approximate a tensor $\mathcal{X}\in \mathbb{C}^{I_1 \times I_2\times \cdots \times I_d}$ as the mode products of a smaller tensor and factor matrices, that is,
\begin{equation*}\label{Tucker}
	\mathcal{X} = \mathcal{G}\times _1 \mathbf{Q}^{(1)} \times _2 \mathbf{Q}^{(2)}\times_3 \cdots \times_d \mathbf{Q}^{(d)}, 
\end{equation*}
where $\mathcal{G} \in \mathbb{C}^{R_1\times R_2 \times\cdots \times R_d}$ is the Tucker core tensor, $\mathbf{Q}^{(k)} \in \mathbb{C}^{R_k\times I_k}$ with $R_k \le I_k$ for $k=1,2,\ldots,d$ are full-rank factor matrices. Its graphical illustrations are shown in Fig.~\ref{Graph: Tucker}. 
\begin{figure}[htbp!]
	\includegraphics[width=0.39\textwidth]{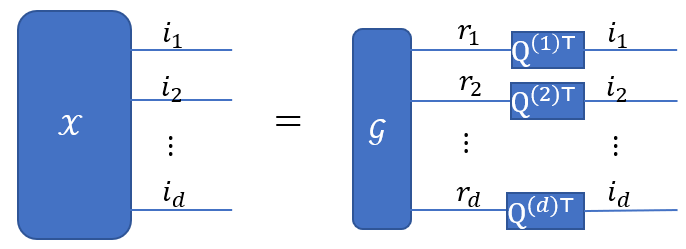} 
	\centering
	\caption{Graphical illustration of Tucker decomposition for a $d$th-order tensor.}
	\label{Graph: Tucker}
\end{figure}

Tucker decomposition is useful for dimensionality reduction and feature extraction of high-dimensional data. Robust algorithms, including truncated HOSVD (THOSVD) \cite{DeLathauwer20001253} and sequentially truncated HOSVD (STHOSVD) \cite{Vannieuwenhoven2012A1027}, have been developed to approximate the low-rank Tucker decomposition. In recent years, randomized algorithms \cite{che2019randomized, che2023efficient, minster2020randomized} have gained much attention due to their effectiveness and flexibility; for example, a randomized version of STHOSVD proposed in \cite{che2023efficient} is shown in Algorithm~\ref{alg: R-STHOSVD}.

\begin{algorithm}[!htbp]
	\renewcommand{\algorithmicrequire}{\textbf{Input:}}
	\renewcommand{\algorithmicensure}{\textbf{Output:}}
	\caption{Randomized STHOSVD algorithm \cite{che2023efficient}}
	\label{alg: R-STHOSVD}
	\centering
	\begin{algorithmic}[1]
		\REQUIRE A tensor $\mathcal{X} \in \mathbb{R}^{I_1 \times I_2 \times \cdots \times I_d}$, target multilinear ranks $\left\{R_1, R_2, \ldots, R_d\right\}$, oversampling
		parameter $K$ and integer $q \geq 1$.
		\ENSURE $N$ orthogonal matrix $\mathbf{Q}^{(k)}$ and core tensor $\mathcal{G}$ such that $\mathcal{X} \approx \mathcal{G} \times_1 \mathbf{Q}^{(1)} \times_2 \mathbf{Q}^{(2)} \times_3 \cdots \times_d \mathbf{Q}^{(d)}$.
		
		\STATE For $k=1,2, \ldots, d$ do:
		\STATE Generate a Gaussian random matrix $\mathbf{G}_k \in \mathbb{R}^{I_k \times\left(R_k+K\right)}$.
		\STATE Form the mode-$k$ unfolding matrix $\mathbf{X}_{(k)}$ of $\mathcal{X}$.
		\STATE Calculate $\mathbf{C}_k=\left(\mathbf{X}_{(k)} \mathbf{X}_{(k)}^{\top}\right)^q \mathbf{G}_k$
		\STATE QR decomposition of $\mathbf{C}_k$ to get its orthogonal columns $\mathbf{Q}$, and set $\mathbf{Q}^{(k)}=\mathbf{Q}_{:, 1: R_k}$.
		\STATE Calculate $\mathcal{X}=\mathcal{X} \times{ }_k \mathbf{Q}^{(k) \top}$.
		\STATE End for.
		\STATE Make the kernel tensor $\mathcal{G}=\mathcal{X}$.
	\end{algorithmic}  
\end{algorithm} 

\subsection{GMP model identification}
The output complex signal of the GMP model can be expressed as
\begin{equation}\label{GMP}	
	y(t) = \sum_{i=0}^{M_1-1}\sum_{j=0}^{M_2-1}\sum_{p=0}^{P-1} \mathcal{S}_{ijp}\ x(t-i)\left| x(t-j) \right|^p,
\end{equation}
where $x(t)$ is the input complex signal, and $| \cdot |$ denotes the modulus of complex numbers. Next, we introduce two different approaches to identify the model parameters.

Suppose there are $N$ sampling data $\{ x(t), y(t) \}_{t_0}^{t_0+N-1}$ to be used for the identification of the GMP model. Let $\mathbf{y}=\left( y(t_0),y(t_ 0+1),\ldots,y(t_0+N-1) \right)^{\top}\in \mathbb{C}^{N}$, and $\mathcal{X}\in \mathbb{C}^{N\times M_1\times M_2\times P}$ denote the fourth-order tensor formed by the basis functions
\[ \mathcal{X}_{nijp} = x((t_0+n)-i)\left| x((t_0+n)-j) \right|^p. \]
We notice that the output $\mathbf{y}$ of the GMP model \eqref{GMP} can be rewritten as the contraction of $\mathcal{X}$ and the coefficient tensor $\mathcal{S}$ as follows
\[ \mathbf{y}=\mathcal{X}\times _{2,3,4}^{1,2,3} \mathcal{S}, \]
as shown in Fig. \ref{Graph: GMP}.
\begin{figure}[!htbp]
	\centering
	\includegraphics[width=0.36\textwidth]{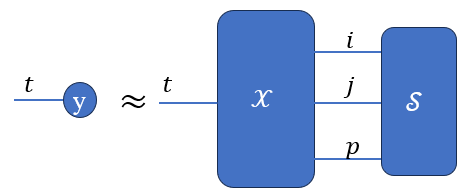}
	\caption{The tensor structure of GMP model}
	\label{Graph: GMP}
\end{figure}

According to the properties of tensor unfolding, the above equation is equivalent to
\[ \mathbf{y} = \mathbf{X}_{(1)} \mathrm{vec}(\mathcal{S}). \]
Hence, the unknown coefficients $\mathcal{S}_{ijp}$'s can be estimated from the ridge regression problem \cite{2014Behavioral, Morgan20063852}
\begin{equation} \label{GMP-identify}
	\min_{\mathbf{s} \in \mathbb{C}^{M_1M_2P}} \left \| \mathbf{y} - \mathbf{X}_{(1)}\mathbf{s} \right \|_F^2 + \gamma \left \| \mathbf{s} \right \|_F^2,  
\end{equation}
where the regularized term $\gamma\| \mathbf{s} \|_F^2$ improves the numerical stability of the problem and can alleviate the overfitting issue by adjusting the penalty parameter $\gamma\ge 0$. Its optimal solution can be directly obtained by the least-squares (LS) method from the following normal equation
\[ (\mathbf{X}_{(1)}^{*}\mathbf{X}_{(1)}+\gamma \mathbf{I})\mathbf{s} = \mathbf{X}_{(1)}^{*}\mathbf{y}, \]
then the coefficient tensor $\mathcal{S} = \mathrm{vec}^{-1}(\mathbf{s})$. 

Another approach is to identify a sparse GMP model with lower complexity from the least absolute shrinkage and selection operator (LASSO) regression problem \cite{wang2023pruning}
\begin{align} \label{GMP-LASSO}
	\min_{\mathbf{s}\in \mathbb{C}^{M_1M_2M_3}}\left \| \mathbf{y} - \mathbf{X}_{(1)}\mathbf{s} \right \|_F^2 + \gamma \left \| \mathbf{s} \right \|_1, 
\end{align}
where the objective function becomes non-differentiable. Hence, the solution can only be solved by iterative methods, such as the proximal gradient descent (PGD) algorithm used in \cite{wang2023pruning}. The fast iterative shrinkage-thresholding algorithm
(FISTA) preserves the computational simplicity of the PGD algorithm but with a better convergence rate \cite{beck2009fast}. The FISTA for \eqref{GMP-LASSO} is given in Algorithm~\ref{alg:FISTA}.

Once the GMP model is identified based on the training data, we compute the normalized mean square error (NMSE) \cite[Chapter 3]{2014Behavioral} 
\begin{equation} \label{NMSE}
	\mathrm{N M S E}\text{(dB)}=10 \log _{10}\left(\frac{\left\|\mathbf{y}_{\text {model }}-\mathbf{y}_{\text {test}}\right\|_F^2}{\left\|\mathbf{y}_{\text {test}}\right\|_F^2}\right)
\end{equation}
on the test dataset to measure the modeling performance, where $\mathbf{y}_{\text {model }}$ denotes the simulated outputs from the trained model, $\mathbf{y}_{\text {test}}$ is the real output signals. 

\begin{algorithm}[!htbp]
	\renewcommand{\algorithmicrequire}{\textbf{Input:}}
	\renewcommand{\algorithmicensure}{\textbf{Output:}}
	\caption{FISTA for GMP model identification via LASSO regression}
	\label{alg:FISTA}
	\centering
	\begin{algorithmic}[1]
		\REQUIRE Training dataset $\{ x(t),y(t) \}_{t_0}^{t_0+N-1}$, memory depths $M_1$ and $M_2$, nonlinearity order $P$, penalty parameter $\gamma \ge 0$, number of iterations $L$, and initial values $\mathbf{s}_0=\mathbf{s}_{-1}$.
		\ENSURE The sparse coefficient tensor $\mathcal{S}$ of GMP model.
		
		\STATE Compute the fourth-order data tensor $\mathcal{X}_{nijp} = x((t_0+n)-i)\left| x((t_0+n)-j) \right|^p$, and its mode-1 unfolding matrix $\mathbf{X}_{(1)} = \mathrm{reshape}(\mathcal{X}, N, M_1M_2P)$.
		\STATE Compute the step size $\alpha = \frac{1}{\| \mathbf{X}_{(1)} \|_2^2}$.
		\STATE For $k=1$ to $L$:
		\STATE Compute $\mathbf{z}_k = \mathbf{s}_{k-1}+\frac{k-2}{k+1}(\mathbf{s}_{k-1}-\mathbf{s}_{k-2})$,
		\STATE Compute $\mathbf{w}_k=\mathbf{z}_k-\alpha \mathbf{X}_{(1)}^*(\mathbf{X}_{(1)}\mathbf{z}_k-\mathbf{y})$,
		\STATE Compute all components of $\mathbf{s}_k$ by
		\[ (\mathbf{s}_k)_i=\left\{\begin{matrix}
			(\mathbf{z}_k)_i(1-\frac{\alpha \gamma}{|(\mathbf{z}_k)_i|}),  & |(\mathbf{z}_k)_i|>\alpha \gamma, \\
			0,\quad  \quad \quad  \quad \quad  \quad & |(\mathbf{z}_k)_i|\le \alpha \gamma,
		\end{matrix}\right. \]
		\STATE End for.
		\STATE Let $\mathcal{S} = \mathrm{reshape}(\mathbf{s}_L, [M_1, M_2, P])$.
	\end{algorithmic}  
\end{algorithm} 

\section{Proposed tensor-based GMP models}
\label{sec. 3}
In this section, we use tensor networks to compress the coefficient tensor $\mathcal{S}$ and propose three tensor-based GMP models, so the potential low-rank tensor structure of $\mathcal{S}$ can be utilized to reduce the number of parameters and model complexity. The ALS methods are derived to identify these models. The model complexity analysis is also presented.

\subsection{GMP-CP model}

In the GMP model \eqref{GMP}, we represent its coefficient tensor $\mathcal{S}$ as the CP format
\[ \mathcal{S}_{i j p}=\sum_{r=1}^R a_{i r} b_{j r} c_{p r}, \]
where $\mathbf{A} \in \mathbb{C}^{M_1 \times R}, \mathbf{B} \in \mathbb{C}^{M_2 \times R}, \mathbf{C} \in \mathbb{C}^{P \times R}$. Then we propose the GMP-CP model
\begin{equation} \label{GMP-CP-original}
	y(t)=\sum_{r=1}^R \sum_{i=0}^{M_1-1} \sum_{j=0}^{M_2-1} \sum_{p=0}^{P-1} a_{i r} b_{j r} c_{p r} x(t-i)|x(t-j)|^p.
\end{equation}

The GMP-CP model has at least three advantages:
\begin{itemize}
\item \textbf{Fewer parameters to characterize the PA systems}. An immediate benefit is the reduction in the number of parameters, as tensor networks provide more compact formats for tensors by making full use of their implicit structure. The number of parameters in GMP-CP model is only $R\left(M_1+M_2+P\right)$, linearly dependent on the memory depths and nonlinearity order, and is usually much smaller than $M_1 M_2 P$ of the GMP model. In our experiments, a small value of $R$ can achieve comparable performance to the GMP model.

\item \textbf{Lower model complexity}. In the mathematical expression \eqref{GMP-CP-original}, we observe that the summation of index $i$ can be carried out separately from the summation of $j,p$ to save the computations, that is,
\begin{equation} \label{GMP-CP}
	y(t) = \sum_{r=1}^{R}\sum_{i=0}^{M_1-1}a_{ir}x(t-i)\sum_{j=0}^{M_2-1}\sum_{p=0}^{P-1} b_{jr} c_{pr} \left| x(t-j)  \right|^p.
\end{equation}
The concrete model complexity analysis can be seen in Subsection \ref{Sec: complexity}. In experiments, the time to simulate the GMP-CP output signals is far less than the GMP model.

\item \textbf{Lower requirements on training data volume and memory storage}. Though the identification for GMP-CP model \eqref{GMP-CP} becomes a multilinear least-squares problem, we derive the ALS algorithm (see Algorithm~\ref{alg: GMP-CP-identify}) to identify the model parameters, where only three small-scale least-squares subproblems need to be solved in each iteration. That requires fewer training data and lower memory storage for the data tensors.
\end{itemize}
	
Next, we derive the ALS algorithm for the identification of GMP-CP model. As shown in Table~\ref{table: Nomenclature}, for $N$ pieces of measured data $\{x(t), y(t)\}^{t_0+N-1}_{t_0}$, let $\mathbf{H} \in \mathbb{C}^{N \times M_1}$ be a rectangular matrix with elements $h_{n i}=x\left(\left(t_0+n\right)-i\right)$, and $\mathcal{M}\in \mathbb{R}^{N\times M_2 \times P}$ be a real third-order tensor computed from input signals $\mathcal{M}_{n j p}=\left|x\left(\left(t_0+n\right)-j\right)\right|^p$. Writing out (\ref{GMP-CP}) for $t=t_0, t_0+1, \ldots, t_0+N-1$ leads to the following tensor equation
\begin{equation}\label{GMP-CP-vectorform}
	\mathbf{y}=\sum_{r=1}^R \mathbf{H A}_{: r} \odot\left(\mathcal{M} \times_2 \mathbf{B}_{: r} \times_3 \mathbf{C}_{: r}\right),
\end{equation}
which is the vector form of GMP-CP model as illustrated in Fig. \ref{Graph: GMP-CP}.
\begin{figure}[htbp!]
	\centering
	\includegraphics[width=0.48\textwidth]{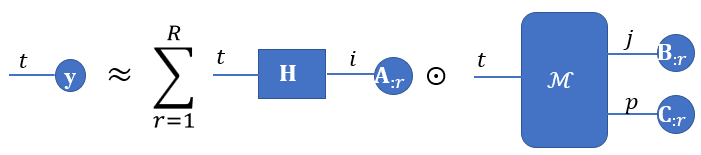}
	\caption{GMP-CP model structure}
	\label{Graph: GMP-CP}
\end{figure}

We identify the GMP-CP model by solving the following regularized multilinear least-squares problem
\begin{align} \label{GMP-CP-identify}
	\min\limits_{\tiny \substack{\mathbf{A} \in \mathbb{C}^{M_1 \times R} \\ \mathbf{B} \in \mathbb{C}^{M_2 \times R} \\ \mathbf{C} \in \mathbb{C}^{P \times R}}}\left\|\mathbf{y}-\sum_{r=1}^R \mathbf{H} \mathbf{A}_{: r} \odot\left(\mathcal{M} \times_2 \mathbf{B}_{: r} \times_3 \mathbf{C}_{: r}\right)\right\|_F^2 \nonumber \\
	+\gamma\left(\|\mathbf{A}\|_F^2+\|\mathbf{B}\|_F^2+\|\mathbf{C}\|_F^2\right),
\end{align}
here `multiliner' means that the equation \eqref{GMP-CP-vectorform} is linear with each of the three parameter matrices $\mathbf{A}, \mathbf{B}, \mathbf{C}$ by fixing the other two. Then the ALS algorithm \cite{Holtz2012701, Rohwedder20131134, Xu20131758} can be adopted to solve the optimization problem \eqref{GMP-CP-identify}. For fixed parameter matrices $\mathbf{B}$ and $\mathbf{C}$, we firstly compute the third-order tensor $\mathcal{D} \in \mathbb{C}^{N\times M_2\times P} $ that
\[ \mathcal{D}_{nir} = h_{ni} (\mathcal{M}\times_2 \mathbf{B}_{:r}\times_3 \mathbf{C}_{:r})_n, \]
then the identification problem \eqref{GMP-CP-identify} boils down to a regularized linear least-squares problem about $\mathbf{A}$,
\[ \min_{\mathbf{A}\in \mathbb{C}^{M_1\times R}} \left \| \mathbf{y} - \mathcal{D}\times _{2,3}^{1,2} \mathbf{A} \right \|^2_F + \gamma \left \| \mathbf{A} \right \|_F^2.  \]
Let $\mathbf{a} = \mathrm{vec}(\mathbf{A})$, the long vector stacked by the columns of $\mathbf{A}$. Then an equivalent problem is yielded as follows
\[ \min_{\mathbf{a}\in \mathbb{C}^{M_1R}} \left \| \mathbf{y} - \mathbf{D}_{(1)} \mathbf{a} \right \|^2_F + \gamma \left \| \mathbf{a} \right \|_F^2,  \]
which is formally similar to the GMP model identification problem \eqref{GMP-identify}, but the number of optimal parameters is greatly reduced. We obtain the solution $\mathbf{a}$ by solving a normal equation and get the updated parameters $\mathbf{A}=\mathrm{vec}^{-1}(\mathbf{a})$; the updating processes of $\mathbf{B}$ and $\mathbf{C}$ are analogous, which can be seen in Algorithm~\ref{alg: GMP-CP-identify}. In each iteration, only three small-scale least-squares problems need to be solved, which greatly reduces the amount of required training data, and realizes rapid updates of model parameters.

\begin{algorithm}[htbp!] \label{ALS-CP model}
	\renewcommand{\algorithmicrequire}{\textbf{Input:}}
	\renewcommand{\algorithmicensure}{\textbf{Output:}}
	\caption{ALS algorithm for the identification of GMP-CP model}
	\label{alg: GMP-CP-identify}
	\centering
	\begin{algorithmic}[1]
		\REQUIRE Training dataset $\{ x(t),y(t) \}_{t_0}^{t_0+N-1}$, memory depths $M_1, M_2$, nonlinearity order $P$, CP-rank $R$, penalty parameter $\gamma\ge 0$, and number of iterations $L$.
		\ENSURE GMP-CP model parameters $\mathbf{A} \in \mathbb{C}^{M_1\times R}$, $\mathbf{B} \in \mathbb{C}^{M_2 \times R}$, $\mathbf{C} \in \mathbb{C}^{P\times R}$.
		
		\STATE Compute $\mathbf{y} = (y(t_0),y(t_0+1),\ldots, y(t_0+N-1))^{\top}$, matrix $\mathbf{H}\in \mathbb{C}^{N\times M_1}$ with entries $h_{ni} = x((t_0+n)-i)$, and third-order tensor $\mathcal{M}\in \mathbb{R}^{N\times M_2 \times P}$ with $\mathcal{M}_{njp} = | x((t_0+n)-j) |^p$. 
		\STATE Set initial values $\mathbf{A}_0\in \mathbb{C}^{M_1\times R}, \mathbf{B}_0\in \mathbb{C}^{M_2\times R}, \mathbf{C}_0\in \mathbb{C}^{P \times R}$.
		\STATE For $k=1$ to $L$: 
		\STATE Compute the third-order tensor $\mathcal{D}_{nir} = h_{ni} (\mathcal{M}\times_2 (\mathbf{B}_{k-1})_{:r}\times_3 (\mathbf{C}_{k-1})_{:r})_n$, and solve the subproblem of $\mathbf{A}$ as follows
		\[\min_{\mathbf{a}\in \mathbb{C}^{M_1R}} \| \mathbf{y}-\mathbf{D}_{(1)}\mathbf{a} \|_F^2 + \gamma \| \mathbf{a} \|_F^2 ,\]
		update $\mathbf{A}_k = \mathrm{reshape}(\mathbf{a}, [M_1,R])$; 
		\STATE Compute the third-order tensor $\mathcal{E}_{njr} = (\mathbf{HA}_k)_{nr} (\mathcal{M}\times_3 (\mathbf{C}_{k-1})_{:r})_{nj}$, and solve the subproblem of $\mathbf{B}$ as follows
		\[\min_{\mathbf{b}\in \mathbb{C}^{M_2R}} \| \mathbf{y}-\mathbf{E}_{(1)}\mathbf{b} \|_F^2 + \gamma \| \mathbf{b} \|_F^2, \]
		update $\mathbf{B}_k = \mathrm{reshape}(\mathbf{b}, [M_2,R])$; 
		\STATE Compute the third-order tensor $\mathcal{F}_{npr} = (\mathbf{HA}_k)_{nr} (\mathcal{M}\times_2 (\mathbf{B}_k)_{:r})_{np}$, and solve the subproblem of $\mathbf{C}$, that is,
		\[ \min_{\mathbf{c}\in \mathbb{C}^{PR}} \| \mathbf{y}-\mathbf{F}_{(1)}\mathbf{c} \|_F^2 + \gamma \| \mathbf{c} \|_F^2, \]
		update $\mathbf{C}_k = \mathrm{reshape}(\mathbf{c}, [P,R])$.
		\STATE End for.
		\STATE Set $\mathbf{A} = \mathbf{A}_L, \mathbf{B} = \mathbf{B}_L, \mathbf{C} = \mathbf{C}_L$.
	\end{algorithmic}  
\end{algorithm}

\subsection{GMP-TT model}

In the same spirit, if we represent the coefficient tensor of the GMP model via the TT format
\[ \mathcal{S}_{i j p}=\sum_{r_1=1}^{R_1} \sum_{r_2=1}^{R_2} a_{i r_1} \mathcal{B}_{r_1 j r_2} c_{r_2 p}, \]
where $\mathbf{A} \in \mathbb{C}^{M_1 \times R_1}, \mathcal{B} \in \mathbb{C}^{R_1 \times M_2 \times R_2}, \mathbf{C} \in \mathbb{C}^{R_2 \times P}$, then we obtain the GMP-TT model
\small
\begin{align}
	y(t) & =\sum_{r_1=1}^{R_1} \sum_{r_2=1}^{R_2} \sum_{i=0}^{M_1-1} \sum_{j=0}^{M_2-1} \sum_{p=0}^{P-1} a_{i r_1} \mathcal{B}_{r_1 j r_2} c_{r_2 p} x(t-i)|x(t-j)|^p \nonumber \\
	& =\sum_{r_1=1}^{R_1} \sum_{i=0}^{M_1-1} a_{i r_1} x(t-i) \sum_{r_2=1}^{R_2} \sum_{j=0}^{M_2-1} \sum_{p=0}^{P-1} \mathcal{B}_{r_1 j r_2} c_{r_2 p}|x(t-j)|^p .
\end{align}
\normalsize
Its parameters are two matrices $\mathbf{A}, \mathbf{C}$ and a third-order tensor $\mathcal{B}$, the total number of parameters $(R_1 M_1+R_2 P+R_1 R_2 M_2)$ only linearly depends on the memory depths and nonlinearity order as well, and depends on the required TT-ranks $(R_1, R_2)$ in practical applications. Lower TT-ranks of the potential coefficient tensor $\mathcal{S}$ imply lower parameter count and model complexity. The vector form of the GMP-TT model is 
\begin{equation}\label{GMP-TT-vector}
	\mathbf{y}=\sum_{r_1=1}^{R_1} \mathbf{H} \mathbf{A}_{: r_1} \odot \sum_{r_2=1}^{R_2}\mathcal{M} \times_2 \mathcal{B}_{r_1: r_2} \times_3 \mathbf{C}_{r_2:},
\end{equation}
the graphic illustration is shown in Fig.~\ref{fig: GMP-TT}.
\begin{algorithm}[!htbp] 
	\renewcommand{\algorithmicrequire}{\textbf{Input:}}
	\renewcommand{\algorithmicensure}{\textbf{Output:}}
	\caption{ALS algorithm for the identification of GMP-TT model}
	\label{alg: GMP-TT-identify}
	\centering
	\begin{algorithmic}[1]
		\REQUIRE Training dataset $\{ x(t),y(t) \}_{t_0}^{t_0+N-1}$, memory depths $M_1, M_2$, nonlinearity order $P$, TT-ranks $(R_1,R_2)$, penalty parameter $\gamma\ge 0$, and number of iterations $L$.
		\ENSURE GMP-TT model parameters $\mathbf{A} \in \mathbb{C}^{M_1\times R_1}, \mathcal{B} \in \mathbb{C}^{R_1\times M_2\times R_2}, \mathbf{C}\in \mathbb{C}^{R_2\times P}$.
		
		\STATE Compute output vector $\mathbf{y} \in \mathbb{C}^{N}$, matrix $\mathbf{H}\in \mathbb{C}^{N\times M_1}$, and third-order tensor $\mathcal{M}\in \mathbb{C}^{N\times M_2 \times P}$. 
		\STATE Set initial values $\mathbf{A}_0\in \mathbb{C}^{M_1\times R_1}$, $\mathcal{B}_0\in \mathbb{C}^{R_1 \times M_2 \times R_2}$, $\mathbf{C}_0\in \mathbb{C}^{R_2 \times P}$.
		\STATE For $k=1$ to $L$: 
		\STATE Calculate the third-order tensor $\mathcal{D}_{nir_1} = h_{ni} \sum_{r_2=1}^{R_2}(\mathcal{M}\times_2 (\mathcal{B}_{k-1})_{r_1 : r_2}\times_3 (\mathbf{C}_{k-1})_{r_2:})_n$, and solve the subproblem of $\mathbf{A}$ as follows
		\[ \min_{\mathbf{a}\in \mathbb{C}^{M_1R_1}} \| \mathbf{y}-\mathbf{D}_{(1)}\mathbf{a} \|_F^2 + \gamma \| \mathbf{a} \|^2_F, \]
		update $\mathbf{A}_k = \mathrm{reshape}(\mathbf{a}, [M_1,R_1])$; 
		\STATE Calculate the fourth-order tensor $\mathcal{E}_{nr_1jr_2} = (\mathbf{HA}_k)_{nr_1} (\mathcal{M} \times_3 \mathbf{C}_{k-1})_{njr_2}$, and solve the subproblem of $\mathcal{B}$ as follows
		\[ \min_{\mathbf{b}\in \mathbb{C}^{R_1M_2R_2}} \| \mathbf{y}-\mathbf{E}_{(1)}\mathbf{b} \|_F^2 + \gamma \| \mathbf{b} \|_F^2, \]
		update $\mathcal{B}_k = \mathrm{reshape}(\mathbf{b}, [R_1, M_2, R_2])$.
		\STATE Calculate the third-order tensor $\mathcal{F}_{nr_2p} = \sum_{r_1=1}^{R_1} (\mathbf{HA}_k)_{nr_1}(\mathcal{M} \times_2 (\mathcal{B}_k)_{r_1:r_2})_{np}$, and solve the subproblem of $\mathbf{C}$, that is,
		\[ \min_{\mathbf{c}\in \mathbb{C}^{R_2P}} \| \mathbf{y}-\mathbf{F}_{(1)}\mathbf{c} \|_F^2 + \gamma \| \mathbf{c} \|_F^2, \]
		update $\mathbf{C}_k = \mathrm{rshape}(\mathbf{c}, [R_2,P])$.
		\STATE End for.
		\STATE Set $\mathbf{A} = \mathbf{A}_L, \mathcal{B}=\mathcal{B}_L, \mathbf{C}=\mathbf{C}_L$. 
	\end{algorithmic}  
\end{algorithm}

The GMP-TT model owns the analogous advantages with the GMP-CP model \eqref{GMP-CP}, and its identification is also a multilinear least-squares problem with regularized terms
\small
\begin{align} \label{GMP-TT-identify}
	\min _{\tiny \substack{\mathbf{A} \in \mathbb{C}^{M_1 \times R_1} \\ \mathcal{B} \in \mathbb{C}^{R_1 \times M_2 \times R_2} \\ \mathbf{C} \in \mathbb{C}^{R_2 \times P}}}\left\|\mathbf{y}-\sum_{r_1=1}^{R_1} \mathbf{H} \mathbf{A}_{: r_1} \odot \sum_{r_2=1}^{R_2}\mathcal{M} \times_2 \mathcal{B}_{r_1: r_2} \times_3 \mathbf{C}_{r_2:}\right\|_F^2 \nonumber \\
	+\gamma\left(\|\mathbf{A}\|_F^2+\|\mathcal{B}\|_F^2+\|\mathbf{C}\|_F^2\right).
\end{align}
\normalsize
We present the ALS method for problem \eqref{GMP-TT-identify} in Algorithm~\ref{alg: GMP-TT-identify}.

\begin{figure}[htbp!]
	\centering
	\includegraphics[width=0.48\textwidth]{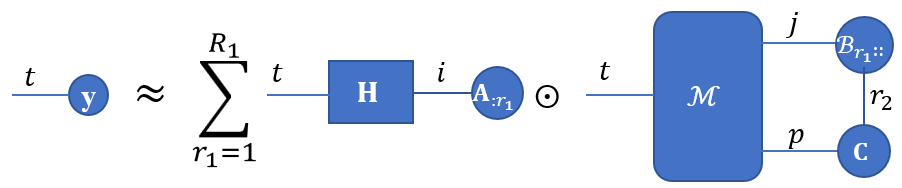}
	\caption{GMP-TT model structure}
	\label{fig: GMP-TT}
\end{figure}

\subsection{GMP-Tucker model}
If we want to exploit the low multilinear-rank structure of the coefficient tensor $\mathcal{S}\in \mathbb{C}^{M_1\times M_2 \times P}$, the Tucker format
\[ \mathcal{S}_{ijp} = \sum_{r_1=1}^{R_1} \sum_{r_2=1}^{R_2}\sum_{r3=1}^{R_3} \mathcal{G}_{r_1r_2r_3} a_{ir_1}b_{jr_2}c_{pr_3} \]
can be used, where $\mathcal{G} \in \mathbf{C}^{R_1 \times R_2 \times R_3}$ is the Tucker core, and $\mathbf{A}\in \mathbb{C}^{M_1\times R_1}, \mathbf{B}\in \mathbb{C}^{M_2 \times R_2}, \mathbf{C}\in \mathbb{C}^{P\times R_3}$ are factor matrices. Then we propose the GMP-Tucker model as follows
\begin{equation}
	\begin{split} \label{GMP-Tucker}
		y(t) = \sum_{r_1,r_2,r_3=1}^{R_1,R_2,R_3}&\sum_{i=0}^{M_1-1} \sum_{j=0}^{M_2-1}\sum_{p=0}^{P-1}  \mathcal{G}_{r_1r_2r_3}a_{ir_1}b_{jr_2}c_{pr_3}
		\\
		&x(t-i)|x(t-j)|^p.
	\end{split}
\end{equation}
The model parameters are three matrices and a core tensor, whose parameter count is  $R_1R_2R_3+M_1R_1+M_2R_2+PR_3$. The vector format of \eqref{GMP-Tucker} can be written as
\begin{align}\label{GMP-Tucker-vector}
    \mathbf{y}=\sum_{r_1,r_2,r_3=1}^{R_1,R_2,R_3}\mathcal{G}_{r_1r_2r_3}\mathbf{H}\mathbf{A}_{:r_1}\odot \mathcal{M}\times _2\mathbf{B}_{:r_2}\times _3 \mathbf{C}_{:r_3}.
\end{align}
Fig.~\ref{fig: GMP-tuckermodel} presents the structure of the GMP-Tucker model.
\begin{figure}[htbp!]
	\centering
	\includegraphics[width=0.39\textwidth]{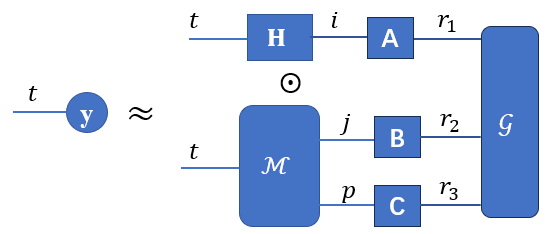}
	\caption{GMP-Tucker model structure}
    \label{fig: GMP-tuckermodel}
\end{figure}

To identify these unknown parameters, we solve a regularized multilinear least-squares problem
\begin{align} 
	\min_{\tiny \substack{\mathcal{G},  \mathbf{A}, \\ \mathbf{B}, \mathbf{C}}} \left\|\mathbf{y}-\sum_{r_1,r_2,r_3=1}^{R_1,R_2,R_3}\mathcal{G}_{r_1r_2r_3} \mathbf{H}\mathbf{A}_{:r_1}\odot \mathcal{M}\times _2\mathbf{B}_{:r_2}\times _3 \mathbf{C}_{:r_3}\right\|_F^2 \nonumber \\
	+\gamma\left(\| \mathcal{G} \|_F^2+ \|\mathbf{A}\|_F^2+\|\mathbf{B}\|_F^2+\|\mathbf{C}\|_F^2\right)
\end{align}
by ALS method as shown in Algorithm~\ref{alg: GMP-Tucker-identify}. In each iteration, we firstly update the core tensor and then three factor matrices sequentially.

\begin{algorithm}[htbp!]
	\renewcommand{\algorithmicrequire}{\textbf{Input:}}
	\renewcommand{\algorithmicensure}{\textbf{Output:}}
	\caption{ALS algorithm for the identification of GMP-Tucker model}
	\label{alg: GMP-Tucker-identify}
	\centering
	\begin{algorithmic}[1]
		\REQUIRE Training dataset $\{ x(t),y(t) \}_{t_0}^{t_0+N-1}$, memory depths $M_1, M_2$ and nonlinearity order $P$, multilinear ranks $(R_1,R_2,R_3)$, number of iterations $L$, penalty parameter $\gamma\ge 0$.
		\ENSURE GMP-Tucker model parameters $\mathcal{G}\in \mathbb{C}^{R_1\times R_2 \times R_3}$, $\mathbf{A} \in \mathbb{C}^{M_1\times R_1}$, $\mathbf{B} \in \mathbb{C}^{M_2 \times R_2}$, $\mathbf{C} \in \mathbb{C}^{P\times R_3}$.
		
		\STATE Compute the output vector $\mathbf{y}\in \mathbb{C}^{N}$, matrix $\mathbf{H}\in \mathbb{C}^{N\times M_1}$, and third-order tensor $\mathcal{M}\in \mathbb{C}^{N\times M_2 \times P}$. 
		\STATE Set initial values $\mathcal{G}_0\in \mathbb{C}^{R_1\times R_2 \times R_3}$, $\mathbf{A}_0\in \mathbb{C}^{M_1\times R_1}, \mathbf{B}_0\in \mathbb{C}^{M_2\times R_2}, \mathbf{C}_0\in \mathbb{C}^{P \times R_3}$.
		\STATE For $k=1$ to $L$: 
		\STATE Compute the fourth-order tensor $\mathcal{K}_{nr_1r_2r_3} = (\mathbf{H}\mathbf{A}_{k-1})_{nr_1} (\mathcal{M}\times_2 \mathbf{B}_{k-1}^{\top}\times_3 \mathbf{C}_{k-1}^{\top})_{nr_2r_3}$, and solve $\mathbf{g}\in \mathbb{C}^{R_1 R_2  R_3}$ from the least-squares problem
			$\min_{\mathbf{g}\in \mathbb{C}^{R_1R_2R_3}} \| \mathbf{y}-\mathbf{K}_{(1)}\mathbf{g} \|_F^2 + \gamma \| \mathbf{g} \|_F^2 ,$
		then update $\mathcal{G}_k = \mathrm{reshape}(\mathbf{g}, [R_1,R_2,R_3])$; 
		\STATE Compute the third-order tensor $\mathcal{D}_{nir_1} = \sum_{r_2,r_3} (\mathcal{G}_{k})_{r_1r_2r_3}h_{ni} (\mathcal{M}\times_2 \mathbf{B}_{k-1}^{\top}\times_3 \mathbf{C}_{k-1}^{\top})_{nr_2r_3}$, and solve $\mathbf{a}\in \mathbb{C}^{M_1 R_1}$ from the least-squares problem
		$\min_{\mathbf{a}\in \mathbb{C}^{M_1R_1}} \| \mathbf{y}-\mathbf{D}_{(1)}\mathbf{a} \|_F^2 + \gamma \| \mathbf{a} \|_F^2, $
	then update $\mathbf{A}_k =  \mathrm{reshape}(\mathbf{a}, [M_1,R_1])$; 
		\STATE Compute the third-order tensor $\mathcal{E}_{njr_2} = \sum_{r_1,r_3}(\mathcal{G}_k)_{r_1r_2r_3}(\mathbf{HA}_k)_{nr_1} (\mathcal{M}\times_3 \mathbf{C}^{\top}_{k-1})_{njr_3}$, and solve $\mathbf{b}\in \mathbb{C}^{M_2 R_2}$ from the least-squares problem $\min_{\mathbf{b}\in \mathbb{C}^{M_2R_2}} \| \mathbf{y}-\mathbf{E}_{(1)}\mathbf{b} \|_F^2 + \gamma \| \mathbf{b} \|_F^2, $
	then update $\mathbf{B}_k =  \mathrm{reshape}(\mathbf{b}, [M_2,R_2])$; 
		\STATE Compute the third-order tensor $\mathcal{F}_{npr_3} = \sum_{r_1,r_2}(\mathcal{G}_k)_{r_1r_2r_3}(\mathbf{HA}_k)_{nr_1} (\mathcal{M}\times_2\mathbf{B}_k^{\top})_{nr_2p}$, and solve $\mathbf{c}\in \mathbb{C}^{P R_3}$ from the least-squares problem
		$\min_{\mathbf{c}\in \mathbb{C}^{PR_3}} \| \mathbf{y}-\mathbf{F}_{(1)}\mathbf{c} \|_F^2 + \gamma \| \mathbf{c} \|_F^2, $
	then update $\mathbf{C}_k = \mathrm{reshape}(\mathbf{c}, [P,R_3])$.
		\STATE End for.
		\STATE Set $\mathcal{G}=\mathcal{G}_L, \mathbf{A} = \mathbf{A}_L, \mathbf{B} = \mathbf{B}_L, \mathbf{C} = \mathbf{C}_L$.
	\end{algorithmic}  
\end{algorithm}

\subsection{Model complexity analysis}
\label{Sec: complexity}
In this subsection, we refer to the analytic methods in reference \cite{tehrani2010comparative} and analyze the running complexity of tensor-based GMP models. Running complexity is the number of calculations that are done on each sample when the model is utilized. Throughout the article, the so-called model complexity exactly means the running complexity. Referred to the comprehensive study for PA models \cite{tehrani2010comparative}, the number of floating-point operations (FLOPs) is a convincing measure for model complexity, and the real addition or multiplication needs 1 FLOP, the complex addition or complex-real multiplication requires 2 FLOPs, the computations for a complex-complex multiplication need 6 FLOPs, and the modulus $|\cdot |$ for a complex number requires 10 FLOPs.

To implement the GMP model, the basis function terms to be generated are $x(t)|x(t-k)|^p$ for each $p$ for all $k=-M_1+1,\cdots, M_2-1$ , other terms can be generated by delaying existing terms. They can be computed by first constructing the $|x(t)|$ with 10 FLOPs, then for fixed $-M_1+1\le k \le M_2-1$ compute the $x(t)|x(t-k)|^p$ for nonlinear order $p$ from 0 to $P-1$ sequentially, hence, 2 FLOPs for each order. Accordingly, these basis terms can be created with $2(P-1)(M_1+M_2-1)+10$ FLOPs. Next, it is to filter the basis terms with coefficients $\mathcal{S}_{ijp}$, which needs $8M_1M_2P-2$ FLOPs. Therefore, the model complexity for GMP model \eqref{GMP} is $ 2(P-1)\left(M_1+M_2-1\right)+8 M_1 M_2 P+8$ FLOPs.

As for tensor-based GMP models, thanks to the representation of tensor networks, the summations on indicators $i$ and $j,p$ are separated, hence, the basis functions $x(t)$ and $|x(t-k)|^p$ can be constructed separately to save many computations. For the GMP-CP model \eqref{GMP-CP}, the basis functions to calculate are merely the real terms $|x(t)|^p$ for each $p$, which need $P+8$ FLOPs. Then for each $r=1,2,\ldots,R$, filtering the basis $x(t)$ and coefficients $a_{ir}$, i.e., the computation for $\sum_{i=0}^{M_1-1}a_{ir}x(t-i)$ requires $8M_1-2$ FLOPs, and the computation for filtering $\sum_{j=0}^{M_2-1}\sum_{p=0}^{P-1} b_{jr} c_{pr} \left| x(t-j)  \right|^p$ needs $10M_2P-2$ FLOPs. Next, computing the complex-complex multiplication for the above summations needs 6 FLOPs. Thereby, for each $r$, there are $10M_2P+8M_1+2$ FLOPs. Finally, the summation with respect to the outermost variable $r$ requires $2R-2$ FLOPs. The whole process needs $R(10M_2P+8M_1+4)+P+6$ FLOPs. 

Similarly, we can analyze the running complexity of GMP-TT and GMP-Tucker models, as shown in Table~\ref{table: complexity}.  

\begin{table}[htbp!]
	\centering
	\caption{The model running complexity}
	\begin{tabular}{|l|l|}
		\hline Model & FLOPs \\
		\hline GMP & $8 M_1 M_2 P+2(P-1)\left(M_1+M_2-1\right)+8$ \\
		\hline GMP-CP & $R\left(10 M_2 P+8 M_1+4\right)+P+6$ \\
		\hline GMP-TT & $R_1(10R_2M_2P+8M_1+4)+P+6$ \\
		\hline GMP-Tucker & $R_1\left(R_2R_3(10M_2P+6)+8M_1+4 \right)+P+6$ \\
		\hline
	\end{tabular}
	\label{table: complexity}
\end{table}

\section{Model reduction via random projections}
\label{sec. 4}
In this section, we reveal the numerical low multilinear-rank structure of the data tensor $\mathcal{M}$, and use random projections computed from its truncated HOSVD to embed the parametric vectors into lower-dimensional spaces, to accelerate the process of model identifications.

\subsection{Random projections for tensor-based GMP models}
In previous tensor-based GMP models, the horizontal slices $\mathcal{M}_{n::}$ of the third-order tensor $\mathcal{M}\in \mathbb{R}^{N\times M_2\times P}$ are Vandermonde matrices
\small
\begin{equation*}
	\begin{bmatrix}
		1 & |x(t_0+n)|  & \cdots & |x(t_0+n)|^P \\
		1 & |x(t_0+n-1)| & \cdots & |x(t_0+n-1)|^P \\
		\vdots & \vdots & \ddots & \vdots \\
		1 & \left|x\left((t_0+n)-M_2+1\right)\right| & \cdots & \left|x\left((t_0+n)-M_2+1\right)\right|^P
	\end{bmatrix},
\end{equation*}
\normalsize
which have a numerically low-rank structure. Therefore, it is possible to reduce the dimensions of $\mathcal{M}$ on modes-$2,3$ by random projections while preserving most of its information. 

Hence, we firstly compute the truncated HOSVD of $\mathcal{M}$ only on modes-2,3 as follows
\begin{equation} \label{Tucker_appro}
	\min _{\tiny \substack{\mathbf{U}_2 \in \mathbb{R}^{ M_2 \times \widetilde{M}_2}, \mathbf{U}_2^{\top} \mathbf{U}_2=\mathbf{I} \\ \mathbf{U}_3 \in \mathbb{R}^{P \times \widetilde{P}}, \mathbf{U}_3^{\top} \mathbf{U}_3=\mathbf{I}}}\left\|\mathcal{M}-\mathcal{M} \times{ }_2 \mathbf{U}_2 \mathbf{U}_2^{\top} \times{ }_3 \mathbf{U}_3 \mathbf{U}_3^{\top}\right\|_F^2,
\end{equation}
where $\widetilde{M}_2\le M_2, \widetilde{P}\le P$. Note that we do not consider the projection on the first mode of $\mathcal{M}$ since the time dimension needs to be consistent with the output vector $\mathbf{y} \in \mathbb{C}^{N}$. The THOSVD \cite{DeLathauwer20001253} and the STHOSVD \cite{Vannieuwenhoven2012A1027} can obtain quasi-optimal solutions for the above problem, but the classical SVDs on the unfolding matrices require uneconomical computational cost. Fortunately, recent randomized algorithms such as the one shown in Algorithm~\ref{alg: R-STHOSVD}, which avoids the SVD processes of unfolding matrices, can achieve high computational efficiency \cite{che2023efficient}.

After the truncated HOSVD of $\mathcal{M}$ to obtain the random matrices $\mathbf{U}_2 \in \mathbb{R}^{M_2 \times \widetilde{M}_2}, \mathbf{U}_3 \in \mathbb{R}^{P \times \widetilde{P}}$ and the core tensor $\widetilde{\mathcal{M}}=\mathcal{M}\times_2 \mathbf{U}_2^{\top}\times _3 \mathbf{U}_3^{\top} \in$ $\mathbb{R}^{N \times \widetilde{M}_2 \times \widetilde{P}}$, we have
\[
\mathcal{M} \approx \mathcal{M} \times_2 \mathbf{U}_2 \mathbf{U}_2^{\top} \times_3 \mathbf{U}_3 \mathbf{U}_3^{\top}=: \widetilde{\mathcal{M}} \times_2 \mathbf{U}_2 \times \mathbf{U}_3,
\]
then it can yield low-dimensional approximations for the tensor-based GMP models by these random projections. 

For example, as for the GMP-CP model \eqref{GMP-CP-vectorform}, it holds that
\begin{align} \label{projected-GMP-CP}
	\mathbf{y} & =\sum_{r=1}^R \mathbf{H} \mathbf{A}_{: r} \odot\left(\mathcal{M} \times_2 \mathbf{B}_{: r} \times_3 \mathbf{C}_{: r}\right) \nonumber \\
	& \approx \sum_{r=1}^R \mathbf{H} \mathbf{A}_{: r} \odot\left(\widetilde{\mathcal{M}} \times_2 \mathbf{U}_2 \times_3 \mathbf{U}_3 \times_2 \mathbf{B}_{: r} \times_3 \mathbf{C}_{: r}\right) \nonumber \\
	& =\sum_{r=1}^R \mathbf{H} \mathbf{A}_{: r} \odot\left(\widetilde{\mathcal{M}} \times_2 \widetilde{\mathbf{B}}_{: r} \times_3 \widetilde{\mathbf{C}}_{: r}\right),
\end{align}
where $\widetilde{\mathbf{B}}_{: r}=\mathbf{U}_2^{\top} \mathbf{B}_{: r} \in \mathbb{C}^{\widetilde{M}_2}$, $\widetilde{\mathbf{C}}_{: r}=\mathbf{U}_3^{\top} \mathbf{C}_{: r} \in \mathbb{C}^{\widetilde{P}}$ are the projected vectors. It follows that the identification of the GMP-CP model \eqref{GMP-CP-identify} can be approximately substituted by the following smaller-scale multilinear least-squares problem
\begin{align}\label{GMP-CP-projection-identify}
	\min _{\tiny \substack{\mathbf{A} \in \mathbb{C}^{M_1 \times R} \\ \widetilde{\mathbf{B}} \in \mathbb{C}^{\widetilde{M}_2 \times R} \\ \widetilde{\mathbf{C}} \in \mathbb{C}^{\widetilde{P} \times R}}}\left\|\mathbf{y}-\sum_{r=1}^R \mathbf{H} \mathbf{A}_{: r} \odot\left(\widetilde{\mathcal{M}} \times_2 \widetilde{\mathbf{B}}_{: r} \times_3 \widetilde{\mathbf{C}}_{: r}\right)\right\|_F^2
	\nonumber \\
	+\gamma\left(\|\mathbf{A}\|_F^2+\|\widetilde{\mathbf{B}}\|_F^2+\|\widetilde{\mathbf{C}}\|_F^2\right),
\end{align}
thus, the computational cost of the subsequent ALS algorithm can be reduced. If we get the solution $\widetilde{\mathbf{B}}\in \mathbb{C}^{\widetilde{M}_2\times R}$ and $\widetilde{\mathbf{C}}\in \mathbb{C}^{\widetilde{P}\times R}$ of problem \eqref{GMP-CP-projection-identify}, then the parameters of the original GMP-CP model can be obtained by back substitutions $\mathbf{B} = \mathbf{U}_2\widetilde{\mathbf{B}}_L$, $\mathbf{C}=\mathbf{U}_3\widetilde{\mathbf{C}}_L$. We call this whole process the random projections-alternating least-squares (RP-ALS) algorithm. The algorithmic framework is described in Algorithm~\ref{alg: GMP-CP-projection-identify}.

The RP-ALS algorithm framework is also applicable to the GMP-TT model and GMP-Tucker model; the computational procedures are completely analogous, so we do not repeat them here.

\begin{remark}
	Note that the RP-ALS algorithm is different from the randomized ALS algorithm proposed in \cite{reynolds2016randomized} for the rank reduction problem of CP decomposition. 
\end{remark}

\begin{algorithm}[htbp!]
	\renewcommand{\algorithmicrequire}{\textbf{Input:}}
	\renewcommand{\algorithmicensure}{\textbf{Output:}}
	\caption{RP-ALS algorithm for GMP-CP model identification}
	\label{alg: GMP-CP-projection-identify}
	\centering
	\begin{algorithmic}[1]
		\REQUIRE Training dataset $\{ x(t),y(t) \}_{t_0}^{t_0+N-1}$, CP-rank $R$, memory lengths $M_1, M_2$, nonlinear order $P$,  projective dimensions $\widetilde{M}_2, \widetilde{P}$, ALS iteration number $L$, and $\gamma \ge 0$.
		\ENSURE GMP-CP model parameters $\mathbf{A}\in \mathbb{C}^{M_1\times R}$, $\mathbf{B}\in \mathbb{C}^{M_2\times R}$, $\mathbf{C}\in \mathbb{C}^{P\times R}$.
		
		\STATE Compute the output vector $\mathbf{y}\in \mathbb{C}^{N}$, matrix $\mathbf{H}\in \mathbb{C}^{N\times M_1}$, and third-order tensor $\mathcal{M}\in \mathbb{C}^{N\times M_2\times P}$. 
		\STATE Employ Algorithm~\ref{alg: R-STHOSVD} to compute the truncated HOSVD \eqref{Tucker_appro} of $\mathcal{M}$ on modes-$2,3$, then obtain the column orthogonal matrices $\mathbf{U}_2\in \mathbb{R}^{M_2\times \widetilde{M}_2}, \mathbf{U}_3\in \mathbb{C}^{P\times \widetilde{P}}$ and the core tensor $\widetilde{\mathcal{M}} \in \mathbb{R}^{N\times \widetilde{M}_2 \times \widetilde{P}}$.
		\STATE Set initial value $\mathbf{A}_0$, $\widetilde{\mathbf{B}}_0$, $\widetilde{\mathbf{C}}_0$, and identify these parameter matrices of the projected GMP-CP model \eqref{projected-GMP-CP} from problem \eqref{GMP-CP-projection-identify} using the ALS method (Algorithm~\ref{alg: GMP-CP-identify}). 
		\STATE Set $\mathbf{A}=\mathbf{A}_L$, $\mathbf{B} = \mathbf{U}_2\widetilde{\mathbf{B}}_L$ and $\mathbf{C}=\mathbf{U}_3\widetilde{\mathbf{C}}_L$.
	\end{algorithmic}  
\end{algorithm} 

\subsection{Approximate error analysis}
Employing the RP-ALS algorithm to tensor-based GMP models actually identifies the projected models via the ALS algorithm; hence, we analyze the upper bound for their approximation error to the original ones.
\begin{theorem}[Error bound for projected GMP-CP model]
	\label{Thm: GMP-CP}
	Suppose an optimal solution of the GMP-CP model \eqref{GMP-CP-vectorform} is $\left\{\mathbf{A}^{\star}, \mathbf{B}^{\star}, \mathbf{C}^{\star}\right\}$. Then, the optimal value of the projected GMP-CP model \eqref{projected-GMP-CP}  satisfies
	\begin{align*}
		& \min\limits_{\tiny \substack{\tiny  \mathbf{A},  \widetilde{\mathbf{B}}, \widetilde{\mathbf{C}}}}\left\|\mathbf{y}-\sum_{r=1}^R \mathbf{H A}_{: r} \odot\left(\widetilde{\mathcal{M}} \times_2 \widetilde{\mathbf{B}}_{: r} \times_3 \widetilde{\mathbf{C}}_{: r}\right)\right\|_F \\
		& \leq \min _{\tiny \substack{\tiny \mathbf{A}, \mathbf{B}, \mathbf{C}}}\left\|\mathbf{y}-\sum_{r=1}^R \mathbf{H} \mathbf{A}_{: r} \odot\left(\mathcal{M} \times_2 \mathbf{B}_{: r} \times_3 \mathbf{C}_{: r}\right)\right\|_F \\ & \quad +\|\mathcal{M}-\widehat{\mathcal{M}}\|_F \sum_{r=1}^R\left\|\mathbf{H A}_{: r}^{\star}\right\|_{\infty}\left\|\mathbf{B}_{: r}^{\star}\right\|_F\left\|\mathbf{C}_{: r}^{\star}\right\|_F,
	\end{align*}
	where $\widehat{\mathcal{M}}=\widetilde{\mathcal{M}}\times _2 \mathbf{U}_2\times _3\mathbf{U}_3=\mathcal{M} \times_2 \mathbf{U}_2 \mathbf{U}_2^{\top} \times_3 \mathbf{U}_3 \mathbf{U}_3^{\top}$ denotes the truncated HOSVD tensor of $\mathcal{M}$.
\end{theorem}
\begin{proof}
	\small
	\begin{align*}
		& \min\limits_{\tiny \substack{\mathbf{A},  \widetilde{\mathbf{B}}, \widetilde{\mathbf{C}} }} \left\|\mathrm{y}-\sum_{r=1}^R \mathbf{H A}_{:r}\odot \left(\widetilde{\mathcal{M}} \times_2 \widetilde{\mathbf{B}}_{: r} \times_3 \widetilde{\mathbf{C}}_{: r}\right)\right\|_F \\
		& =\min\limits _{\tiny \substack{\mathbf{A} , \mathbf{B}, \mathbf{C}}}\left\|\mathbf{y}-\sum_{r=1}^R \mathbf{H} \mathbf{A}_{: r} \odot\left(  \widehat{\mathcal{M}} \times_2 \mathbf{B}_{: r} \times_3 \mathbf{C}_{: r}\right)\right\|_F \\
		& \leq \min _{\tiny \substack{\mathbf{A}, \mathbf{B}, \mathbf{C}}}\left\|\mathbf{y}-\sum_{r=1}^R \mathbf{H} \mathbf{A}_{: r} \odot\left(\mathcal{M} \times_2 \mathbf{B}_{: r} \times_3 \mathbf{C}_{: r}\right)\right\|_F \\
		& \quad + \left\|\sum_{r=1}^R \mathbf{H} \mathbf{A}_{: r} \odot\left((\mathcal{M}-\widehat{\mathcal{M}}) \times_2 \mathbf{B}_{: r} \times_3 \mathbf{C}_{: r}\right)\right\|_F \\
		& \leq \text{\rm m(GMP-CP)}+\sum_{r=1}^R\left\|\mathbf{H A}_{: r}^{\star} \odot\left(\mathbf{M}_{(1)}-\widehat{\mathbf{M}}_{(1)}\right)\left(\mathbf{C}_{: r}^{\star} \otimes \mathbf{B}_{: r}^{\star}\right)\right\|_F \\
        &\leq \text{\rm m(GMP-CP)}+\sum_{r=1}^R\left\|\mathbf{H A}_{:r}^{\star}\right\|_{\infty}\left\|(\mathbf{M}_{(1)}-\widehat{\mathbf{M}}_{(1)})\left(\mathbf{C}_{: r}^{\star} \otimes \mathbf{B}_{: r}^{\star}\right)\right\|_F 
\end{align*}
\begin{align*}
	&\leq \text{\rm m(GMP-CP)}+\|\mathbf{M}_{(1)}-\widehat{\mathbf{M}}_{(1)}\|_2\sum_{r=1}^R\left\|\mathbf{H A}_{:r}^{\star}\right\|_{\infty}\left\|\mathbf{C}_{: r}^{\star}\right\|_F\left\|\mathbf{B}_{: r}^{\star}\right\|_F \\
	& \leq \text{\rm m(GMP-CP)}+\|\mathcal{M}-\widehat{\mathcal{M}}\|_F \sum_{r=1}^R\left\|\mathbf{H A}_{: r}^{\star}\right\|_{\infty}\left\|\mathbf{B}_{: r}^{\star}\right\|_F\left\|\mathbf{C}_{: r}^{\star}\right\|_F, 
\end{align*}
	\normalsize where m(GMP-CP) denotes the optimal accuracy of the GMP-CP model as the second equation in this theorem, and the last inequality is obtained from the fact that the spectral norm of a matrix is not larger than its Frobenius norm.
\end{proof}
Theorem~\ref{Thm: GMP-CP} signifies that the difference of the optimal accuracy between the projected GMP-CP model \eqref{projected-GMP-CP} and GMP-CP model \eqref{GMP-CP-vectorform} can be controlled by the approximate error of the truncated HOSVD on the data tensor $\mathcal{M}$. Hereby, it provides a theoretical guarantee to identify the GMP-CP model by the RP-ALS algorithm. It also holds for the GMP-TT model and GMP-Tucker model.

As for the GMP-TT model, the first step of the RP-ALS algorithm also generates random projections, and the projected GMP-TT model is
\begin{equation} \label{projected-GMP-TT}
	\mathbf{y}  = \sum_{r_1=1}^{R_1} \mathbf{H A}_{:r_1} \odot \sum_{r_2=1}^{R_2} (\widetilde{\mathcal{M}} \times _2 \widetilde{\mathcal{B}}_{r_1:r_2} \times _3 \widetilde{\mathbf{C}}_{r_2:}),
\end{equation}
where $\widetilde{\mathcal{B}} = \mathcal{B}\times_2 \mathbf{U}_2^{\top}\in \mathbb{C}^{R_1\times \widetilde{M}_1\times R_2}, \widetilde{\mathbf{C}} = \mathbf{U}_3^{\top}\mathbf{C} \in \mathbb{C}^{\widetilde{P}\times R_2}$.
\begin{theorem}[Error bound for projected GMP-TT model]
	Suppose an optimal solution of the GMP-TT model \eqref{GMP-TT-vector} is $\left\{\mathbf{A}^{\star}, \mathcal{B}^{\star}, \mathbf{C}^{\star}\right\}$. Then, the optimal value of the projected GMP-TT model \eqref{projected-GMP-TT} satisfies
	\small
	\begin{align*}
		& \min\limits_{\tiny \substack{\mathbf{A},
				\widetilde{\mathcal{B}},\widetilde{\mathbf{C}}}}\left\|\mathbf{y}-\sum_{r_1,r_2=1}^{R_1,R_2} \mathbf{H A} _{r_1} \odot\left(\widetilde{\mathcal{M}} \times_2 \widetilde{\mathcal{B}}_{r_1: r_2} \times_3 \widetilde{\mathbf{C}}_{r_2:}\right)\right\|_F \\
		& \leq \min\limits_{\tiny \substack{\mathbf{A}, 
				\mathcal{B},\mathbf{C}}}\left\|\mathbf{y}-\sum_{r_1,r_2=1}^{R_1,R_2} \mathbf{H A} _{r_1} \odot\left(\mathcal{M} \times_2 \mathcal{B}_{r_1: r_2} \times_3 \mathbf{C}_{r_2:}\right)\right\|_F \\
		& \quad +\|\mathcal{M}-\widehat{\mathcal{M}}\|_F \sum_{r_1=1}^{R_1} \sum_{r_2=1}^{R_2}\left\|\mathbf{H A}_{: r_1}^{\star}\right\|_{\infty}\left\|\mathcal{B}_{r_1: r_2}^{\star}\right\|_F\left\|\mathbf{C}_{r_2:}^{\star}\right\|_F.
	\end{align*}
\end{theorem}
\begin{proof}
	Analogous to the proof of Theorem~\ref{Thm: GMP-CP}.
\end{proof}

Likewise, if the RP-ALS algorithm is employed to identify the GMP-Tucker model, then the ALS method is actually applied to the following projected one
\begin{align}\label{projected-GMP-Tucker}
	\mathbf{y}=\sum_{r_1,r_2,r_3=1}^{R_1,R_2,R_3}\mathcal{G}_{r_1r_2r_3}\mathbf{H}\mathbf{A}_{:r_1}\odot \widetilde{\mathcal{M}}\times _2\widetilde{\mathbf{B}}_{:r_2}\times _3 \widetilde{\mathbf{C}}_{:r_3},
\end{align}
where $\widetilde{\mathbf{B}} = \mathbf{U}_2^{\top}\mathbf{B}\in \mathbb{C}^{\widetilde{M}_1\times R_2}, \widetilde{\mathbf{C}} = \mathbf{U}_3^{\top}\mathbf{C} \in \mathbb{C}^{\widetilde{P}\times R_3}$.

\begin{theorem}[Error bound for projected GMP-Tucker model]
	Suppose an optimal solution of the GMP-Tucker model \eqref{GMP-Tucker-vector} is $\left\{\mathcal{G}^{\star},\mathbf{A}^{\star}, \mathbf{B}^{\star}, \mathbf{C}^{\star}\right\}$. Then, the optimal value of the projected GMP-Tucker model \eqref{projected-GMP-Tucker} satisfies
	\small
	\begin{align*}
		& \min_{\tiny \substack{\mathcal{G}, \mathbf{A}, \widetilde{\mathbf{B}}, \widetilde{\mathbf{C}}}} \left\|\mathbf{y}-\sum_{r_1,r_2,r_3=1}^{R_1,R_2,R_3}\mathcal{G}_{r_1r_2r_3} \mathbf{H}\mathbf{A}_{:r_1}\odot \widetilde{\mathcal{M}}\times _2\widetilde{\mathbf{B}}_{:r_2}\times _3 \widetilde{\mathbf{C}}_{:r_3}\right\|_F \\
		& \leq \min\limits_{\tiny \substack{\mathcal{G}, \mathbf{A}, \\ \mathbf{B}, \mathbf{C}}}\left\|\mathbf{y}-\sum_{r_1,r_2,r_3=1}^{R_1,R_2,R_3}\mathcal{G}_{r_1r_2r_3} \mathbf{H}\mathbf{A}_{:r_1}\odot \mathcal{M}\times _2\mathbf{B}_{:r_2}\times _3 \mathbf{C}_{:r_3}\right\|_F \\
		& \  +\|\mathcal{M}-\widehat{\mathcal{M}}\|_F \sum_{r_1,r_2,r_3=1}^{R_1,R_2,R_3}|\mathcal{G}_{r_1r_2r_3}| \left\|\mathbf{H A}_{: r_1}^{\star}\right\|_{\infty}\left\|\mathbf{B}_{: r_2}^{\star}\right\|_F\left\|\mathbf{C}_{:r_3}^{\star}\right\|_F.
	\end{align*}
\end{theorem}
\begin{proof}
	Analogous to the proof of Theorem~\ref{Thm: GMP-CP}.
\end{proof}

\section{Numerical results}
\label{sec. 5}
In this section, we conduct numerical experiments to test the efficiency of our proposed models and algorithms. The experiments are performed on a computer with an AMD Ryzen $7$ $4800 \mathrm{U}$ CPU at $1.80 \mathrm{GHz}$ and 16GB of RAM using MATLAB 2023b. We make use of the open-source toolkit `Tensor Toolbox' \cite{Brett2023} to facilitate the implementation of tensor multiplications.

We simulate the following orthogonal
frequency-division multiplexing (OFDM) system \cite{tellado2003maximum} as the inputs to a PA system. Specifically, the OFDM system with FFT length of 2048, 1584 communication subcarriers, cyclic prefix length of 72, a total of 28 symbols, and modulation mode of $16$ QAM is simulated. The baseband signals in the system are randomly generated by the MATLAB command $\operatorname{randi}([0,1])$. Then, we put the generated input signal $x(t)$ into the PA model $\text{PA(3)}$ in \cite[Table III]{Alina2016603} with memory depth $11$, and add the Gaussian noise of 50 signal-to-noise ratio level to get the nonlinear output signals $y(t)$. We generate 30688 input-output data $\{x(t), y(t)\}_1^{30688}$. Fig. \ref{fig: power} shows the power of this PA system; it is a nonlinear process with memory effect and the nonlinear distortion becomes more serious at higher power.

\begin{figure}[htbp!]
	\centering
	\includegraphics[width=0.38\textwidth]{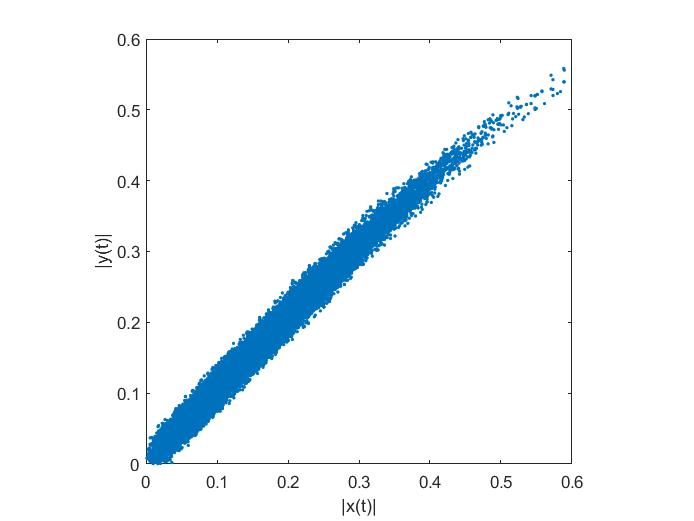}
	\caption{The output power of the simulated PA data.}
	\label{fig: power}
\end{figure}

To ensure the continuity of the training data, we take $N=$ 1024 input-output signals $\{x(t), y(t)\}_{100}^{1123}$ to be used for model identifications, and a total of 30639 data $\{x(t), y(t)\}_{20}^{30668}$ to test the model performance by computing the NMSE \eqref{NMSE}. To make the models achieve better performance, we adjust the penalty parameter $\gamma$ to take appropriate values; the initial values of iterative algorithms for model identification are randomly generated.

We first give some explanations for the following figures and tables that present the numerical results. The term labeled `NumofPara' refers to the number of nonzero parameters, the table column labeled `Training time' indicates the amount of time taken to identify the models, while the column labeled `Simulation time’ shows the computational time for simulating the outputs of testing data. To exclude the influence of randomness, the time and NMSE are averaged over 100 independent experiments. We compare the experimental results of the models and algorithms under two different hyper-parameter settings, concretely, the memory depths and nonlinearity order $(M_1, M_2, P)$ are taken as $(11,10,8)$ and $(10,8,6)$ respectively.

To obtain better performance of the GMP model as the comparative benchmark, we tune the penalty parameter $\gamma$ by a constant step of $10^{-5}$, as shown in Fig. \ref{fig: 1} (a) and (c). For $(M_1, M_2, P)=(11,10,8)$, Fig. \ref{fig: 1} (a) shows that the NMSE of the GMP model via ridge regression \eqref{GMP-identify} under different penalty parameter $\gamma$, the optimal $\gamma=4.8\times 10^{-4}$ and NMSE achieves $-49.3635$ dB. We test the numerical results of  GMP model via LASSO regression under various penalty parameters. For a variety of penalty parameters,
the number of iterations required for the NMSE to reach $-49$ dB is more than $1500$, as shown in Fig. \ref{fig: 1} (b), which presents the NMSEs for GMP model via LASSO  \eqref{GMP-LASSO} with $\gamma = 3\times 10^{-4}$ identified by the FISTA and PGD algorithms, respectively. The FISTA, employed in the following experiments, outperforms the PGD algorithm previously used in \cite{wang2023pruning}. To ensure the model achieves good performance without significantly prolonging the training time due to excessive iterations, we select an appropriate iteration number of 2000 based on multiple tests. Fig. \ref{fig: 1} (c) reports the NMSE performance and model sparsity after $2000$ FISTA iterations under various penalty parameters $\gamma$. To trade off the NMSE and model sparsity, we choose $\gamma=2.7\times 10^{-4}$ such that the sparse model has as few nonzero parameters as possible while maintaining the NMSE comparable to the full GMP model. Fig. \ref{fig: 1} (d) shows the performance of our proposed models identified by ALS algorithm. Their NMSEs decrease rapidly and converge to stable values within several iterations, which is comparable to the GMP model. For $(M_1, M_2, P)=(10,8,6)$, we use the same method as above to tune the penalty parameters. 

\begin{figure}[htbp!]
	\centering
	\subfloat[\tiny]{
		\includegraphics[width=0.24\textwidth]{ 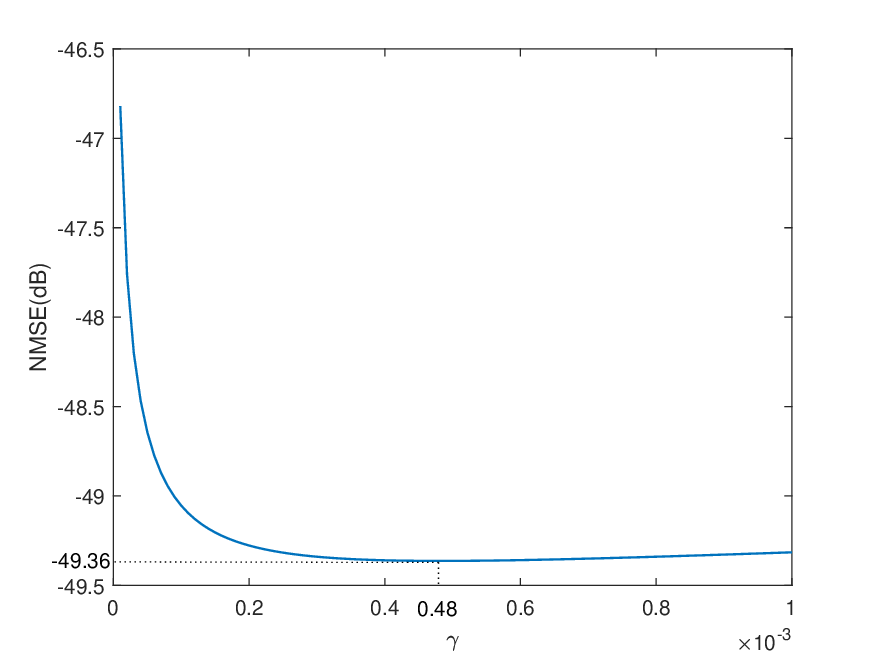}} 
	\subfloat[]{
		\includegraphics[width=0.24\textwidth]{ 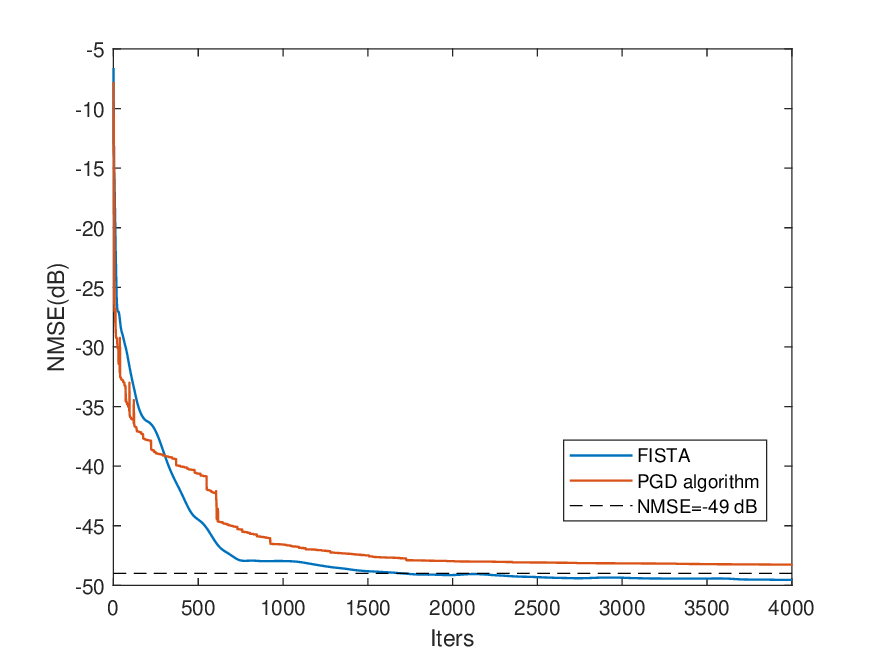}}  \\
	\subfloat[]{
		\includegraphics[width=0.24\textwidth]{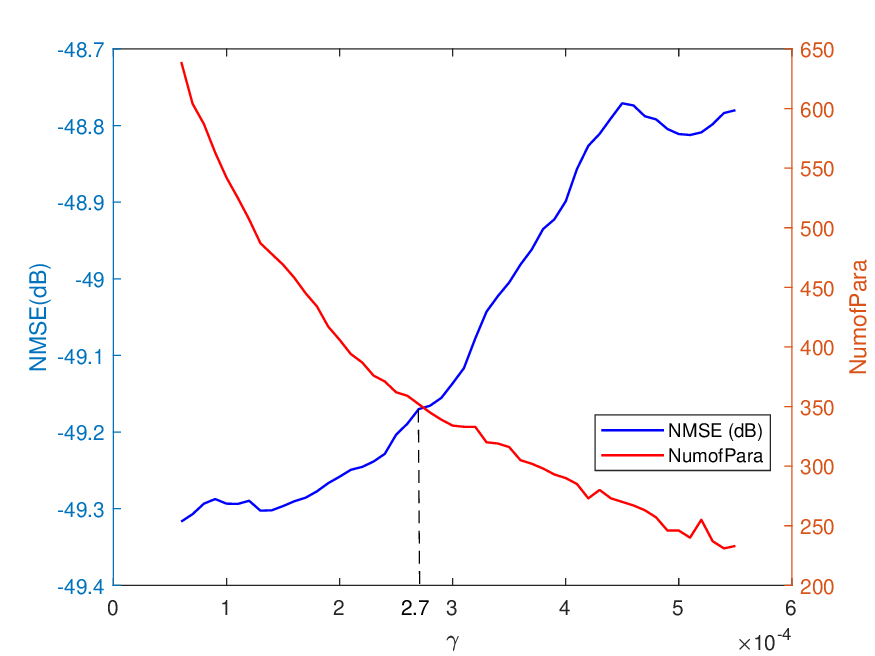}}
	\subfloat[]{
		\includegraphics[width=0.24\textwidth]{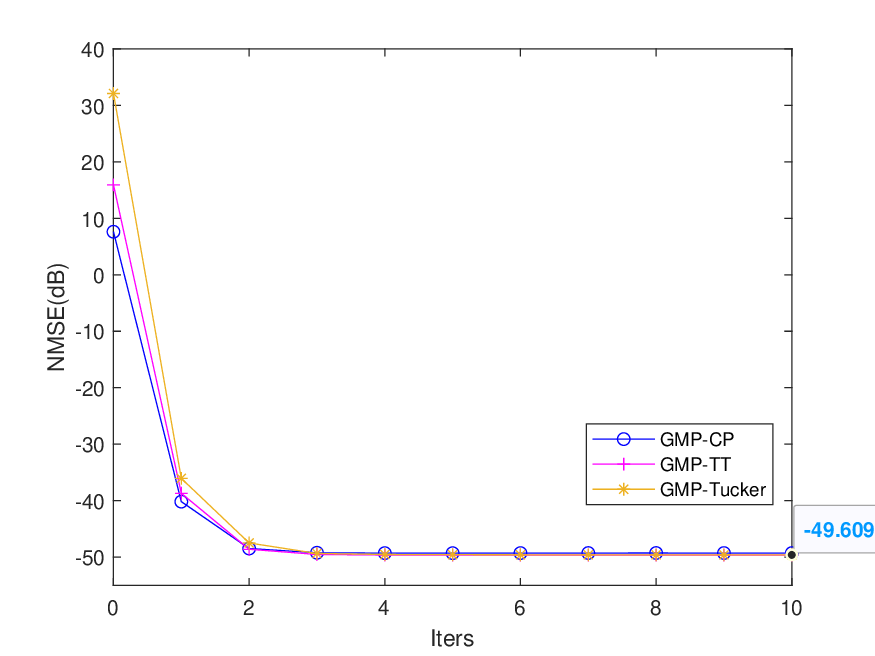}}
	\caption{The hyper-parameters $(M_1,M_2,P)=(11,10,8)$. (a) The NMSE of full GMP model trained by LS method varies with the penalty parameter $\gamma$. (b) The NMSE curve of the sparse GMP model via LASSO  identified by FISTA and PGD algorithm under  $\gamma=3\times 10^{-4}$. (c) Under various penalty parameter $\gamma$, the double y-axis graph for the curves of Numofpara and NMSE of the sparse GMP model identified by FISTA. (d) The NMSE of tensor-based GMP models converge rapidly with ALS iterations, and reach a satisfactory level.}
	\label{fig: 1}
\end{figure}

Table~\ref{table: comparisonALS} shows the concrete statistical data for these models. It can be seen that our proposed tensor-based GMP models can greatly reduce the number of parameters and save time for simulating the output signals, while their NMSEs are comparable to the GMP model. The training time for tensor-based GMP models is the sum of the computation of tensors $\mathcal{M}$ and $\mathbf{H}$, and three ALS iterations. The training time for GMP-CP or GMP-TT models is slightly longer than that for the GMP model via the LS method, and the training time for GMP-Tucker model is relatively longer, but lower than that of the sparse GMP model via LASSO, for which $2000$ and $1400$ FISTA iterations are performed to ensure good model performances in the two experiments, respectively. In some practices, the model identification process can be done offline, but the simulation for the future input signals based on trained models is a real-time process. The tensor-based GMP models exhibit great improvements in terms of the simulation time, especially the GMP-CP model, which offers nearly 9x speedup compared to the full GMP model.
\begin{table*}[htbp!]
	\centering
	\caption{Numerical results of tensor-based GMP models}
	\label{table: comparisonALS}
	\begin{tabular}{|c|c|cccc|cccc|}
		\hline
		\multirow{2}{*}{\begin{tabular}[c]{@{}c@{}}\\ Models\end{tabular}}           & \multirow{2}{*}{\begin{tabular}[c]{@{}c@{}}\\ Tensor\\ ranks\end{tabular}} & \multicolumn{4}{c|}{$\left(M_1, M_2, P\right)=(11,10,8)$}                                                          & \multicolumn{4}{c|}{$\left(M_1, M_2, P\right)=(10,8,6)$}                                                           \\ \cline{3-10} 
		&                                                                         & \multicolumn{1}{c|}{NumofPara} & \multicolumn{1}{c|}{Training time} & \multicolumn{1}{c|}{\begin{tabular}[l]{@{}c@{}}Simulation\\ time\end{tabular}} & NMSE     & \multicolumn{1}{c|}{NumofPara} & \multicolumn{1}{c|}{Training time} & \multicolumn{1}{c|}{\begin{tabular}[l]{@{}c@{}}Simulation\\ time\end{tabular}} & NMSE     \\ \hline
		GMP(LS) & -                                                                       & \multicolumn{1}{c|}{880}       & \multicolumn{1}{c|}{0.1477}        & \multicolumn{1}{c|}{0.7511}       & -49.3635 & \multicolumn{1}{c|}{480}       & \multicolumn{1}{c|}{0.0779}        & \multicolumn{1}{c|}{0.5202}       & -44.4182 \\ \hline
		GMP(LASSO)                        & -                                                                       & \multicolumn{1}{c|}{352}       & \multicolumn{1}{c|}{3.8988}        & \multicolumn{1}{c|}{0.4340}       & -49.1700 & \multicolumn{1}{c|}{$241$}       & \multicolumn{1}{c|}{1.4915}        & \multicolumn{1}{c|}{0.2846}       & -44.2422 \\ \hline
		GMP-CP                            & 3                                                                       & \multicolumn{1}{c|}{87}        & \multicolumn{1}{c|}{0.1763}        & \multicolumn{1}{c|}{0.0801}       & -49.2861 & \multicolumn{1}{c|}{72}        & \multicolumn{1}{c|}{0.1183}        & \multicolumn{1}{c|}{0.0557}       & -44.3296 \\ \hline
		GMP-TT                         & $(2,2)$                                                                 & \multicolumn{1}{c|}{78}   & \multicolumn{1}{c|}{0.2348}        & \multicolumn{1}{c|}{0.1006}       & -49.5338 & \multicolumn{1}{c|}{64}    
		& \multicolumn{1}{c|}{0.1527}        & \multicolumn{1}{c|}{0.0707}       & -44.1102 \\ \hline
		GMP-Tucker                        & $(2,2,2)$                                                               & \multicolumn{1}{c|}{66}   & \multicolumn{1}{c|}{0.6948}        & \multicolumn{1}{c|}{0.1896}       & -49.3578 & \multicolumn{1}{c|}{56}    
		& \multicolumn{1}{c|}{0.5233}        & \multicolumn{1}{c|}{0.1338}       & -44.3128 \\ \hline
	\end{tabular}
\end{table*} 

As demonstrated in Section \ref{sec. 4}, the RP-ALS algorithm is proposed to accelerate the ALS iterations, by embedding the parameter vectors into lower-dimensional spaces via random projections. The more ALS iterations are performed, the more training time is saved. Table~\ref{table: RPALS} records the performance of the RP-ALS algorithm, where the hyper-parameters are consistent with $(M_1, M_2, P)=(11,10,8)$ in Table~\ref{table: comparisonALS}, and the table column `HOSVD' means the running time for the truncated HOSVD on $\mathcal{M}$ for given $(\widetilde{M},\widetilde{P})=(5,3)$, the column named `One ALS iteration' shows the required time for one ALS iteration. We also performed only three ALS iterations, and the NMSE can reach comparable levels. The training time is consumed in the computation of tensors $\mathcal{M},\mathbf{H}$, the randomized STHOSVD process on $\mathcal{M}$, and three ALS iterations. It can be seen that the RP-ALS algorithm needs less time to identify these models than the ALS algorithm, especially for the GMP-Tucker model, since it saves more time per iteration than the other two models. In addition, Fig.~\ref{fig: RP-ALS} (a) shows that the NMSE declines rapidly with 10 ALS iterations. Since the projections $\mathbf{U}_2,\mathbf{U}_3$ are generated from the randomized HOSVD of $\mathcal{M}$, we counted NMSEs after the 10th RP-ALS iteration of 100 independent experiments, and their distributions are plotted in Fig.~\ref{fig: RP-ALS} (b)-(d), which shows that the performances of the RP-ALS method are quite robust.

From the above experiments, our proposed models achieve comparable performances to the original GMP model while greatly reducing the number of parameters and running complexity, and the ALS algorithm and RP-ALS algorithm can efficiently identify these models.

\begin{table}
	\centering
	\caption{Model identification by RP-ALS algorithm}
	\begin{tabular}{|c|c|c|c|c|c|}
		\hline
		Models & \begin{tabular}[c]{@{}c@{}}Tensor\\ ranks\end{tabular} & HOSVD & \begin{tabular}[c]{@{}l@{}}One ALS\\  iteration\end{tabular} & \begin{tabular}[c]{@{}l@{}}Training\\ time\end{tabular} & NMSE     \\ \hline
		GMP-CP  & 3 & 0.0036                                               & 0.0525                                                          & 0.1666                                                  & -49.2822 \\ \hline
		GMP-TT & (2,2)  & 0.0036                                               & 0.0714                                                          & 0.2233                                                  & -49.3685 \\ \hline
		GMP-Tucker & (2,2,2) & 0.0036                                               & 0.1540                                                          & 0.4711                                                  & -49.1266 \\ \hline
	\end{tabular}
	\label{table: RPALS}
\end{table}

\begin{figure}[htbp!]
	\centering
	\subfloat[]{
		\includegraphics[width=0.24\textwidth]{ 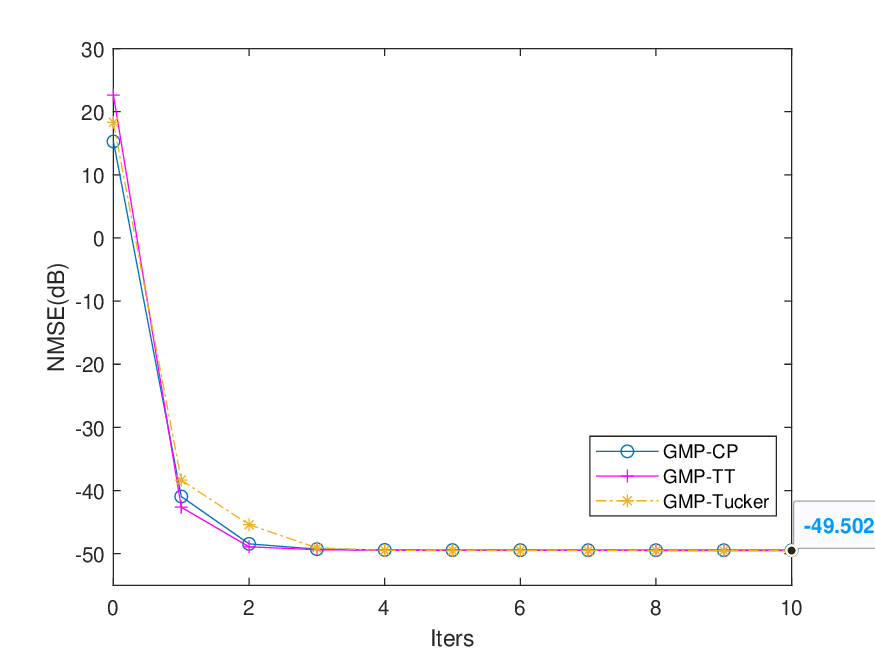}} 
	\subfloat[]{
		\includegraphics[width=0.24\textwidth]{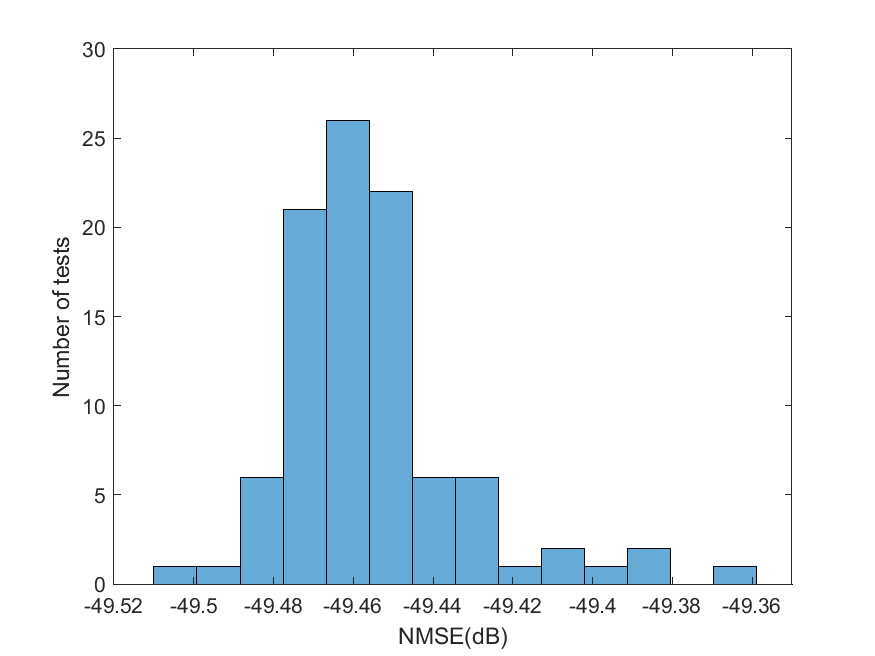}} \\
	\subfloat[]{
		\includegraphics[width=0.24\textwidth]{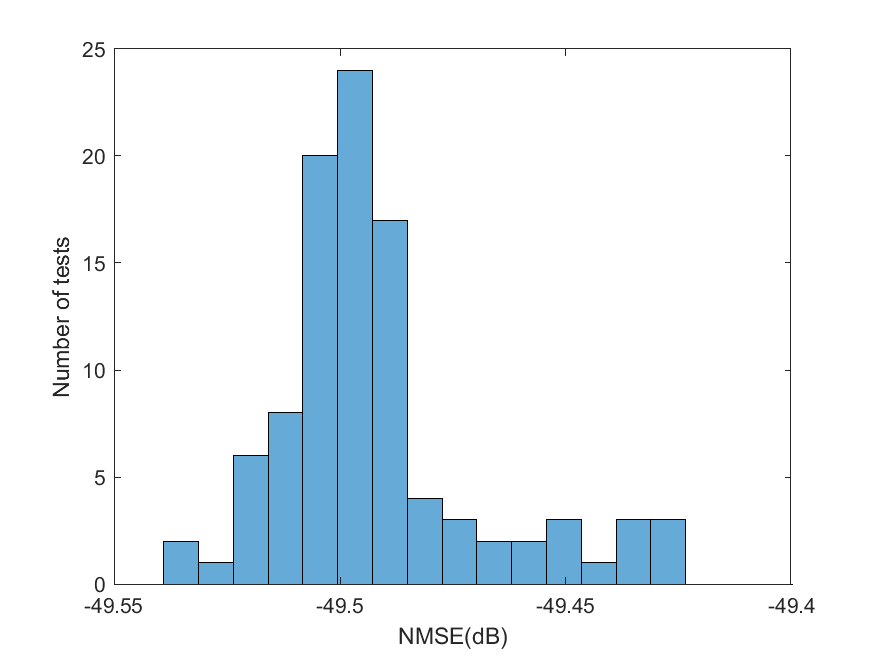}} 
	\subfloat[]{
		\includegraphics[width=0.24\textwidth]{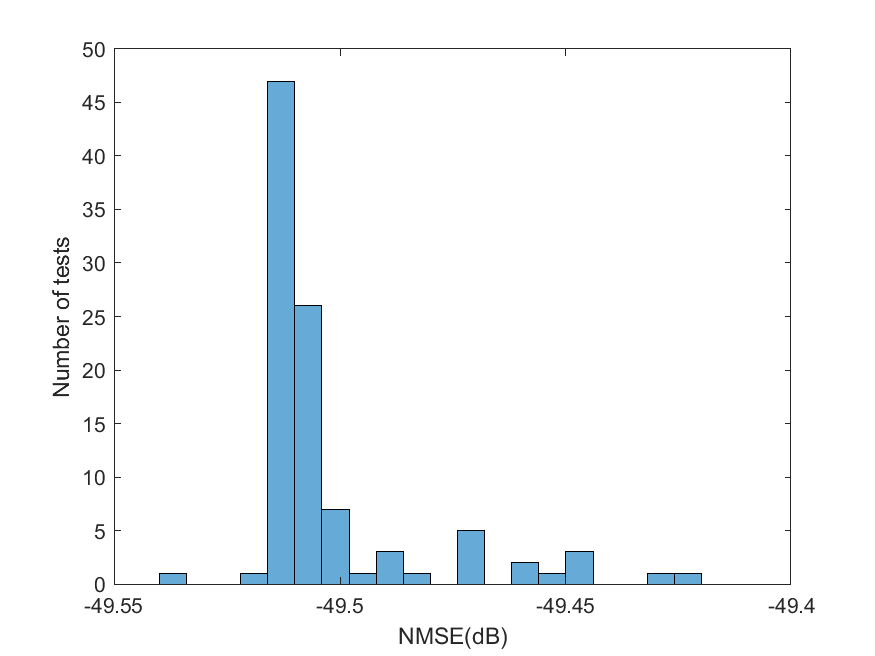}}
	\caption{The hyper-parameters $(M_1,M_2,P)=(11,10,8)$. (a) The NMSE of three tensor-based GMP models trained by RP-ALS algorithm. (b)-(d) The NMSE distributions of GMP-CP, GMP-TT and GMP-Tucker models trained by RP-ALS algorithm, respectively. }
	\label{fig: RP-ALS} 
\end{figure}

\section{Conclusions}
\label{sec. 6}

By utilizing the tensor structure of the unknown coefficients of the GMP model, the proposed tensor-based GMP models can greatly reduce the number of parameters to characterize the PA systems, while maintaining comparable performance. Moreover, the running complexity of these models is much lower than that of the GMP model. We deduce the ALS method to efficiently identify the model parameters in an iterative manner. In addition, we propose the RP-ALS algorithm to accelerate the alternating iterations through a randomized projection technique and confirm the robustness of this algorithm theoretically and experimentally. Numerical results demonstrate the merits of these models in terms of the number of parameters and running complexity.

\section*{Acknowledgements}
The authors would like to thank the handling editor and referees for their detailed comments.

\bibliographystyle{plain}
	
\bibliography{ref}

\begin{IEEEbiography}[{\includegraphics[width=1in,height=1.25in,clip,keepaspectratio]{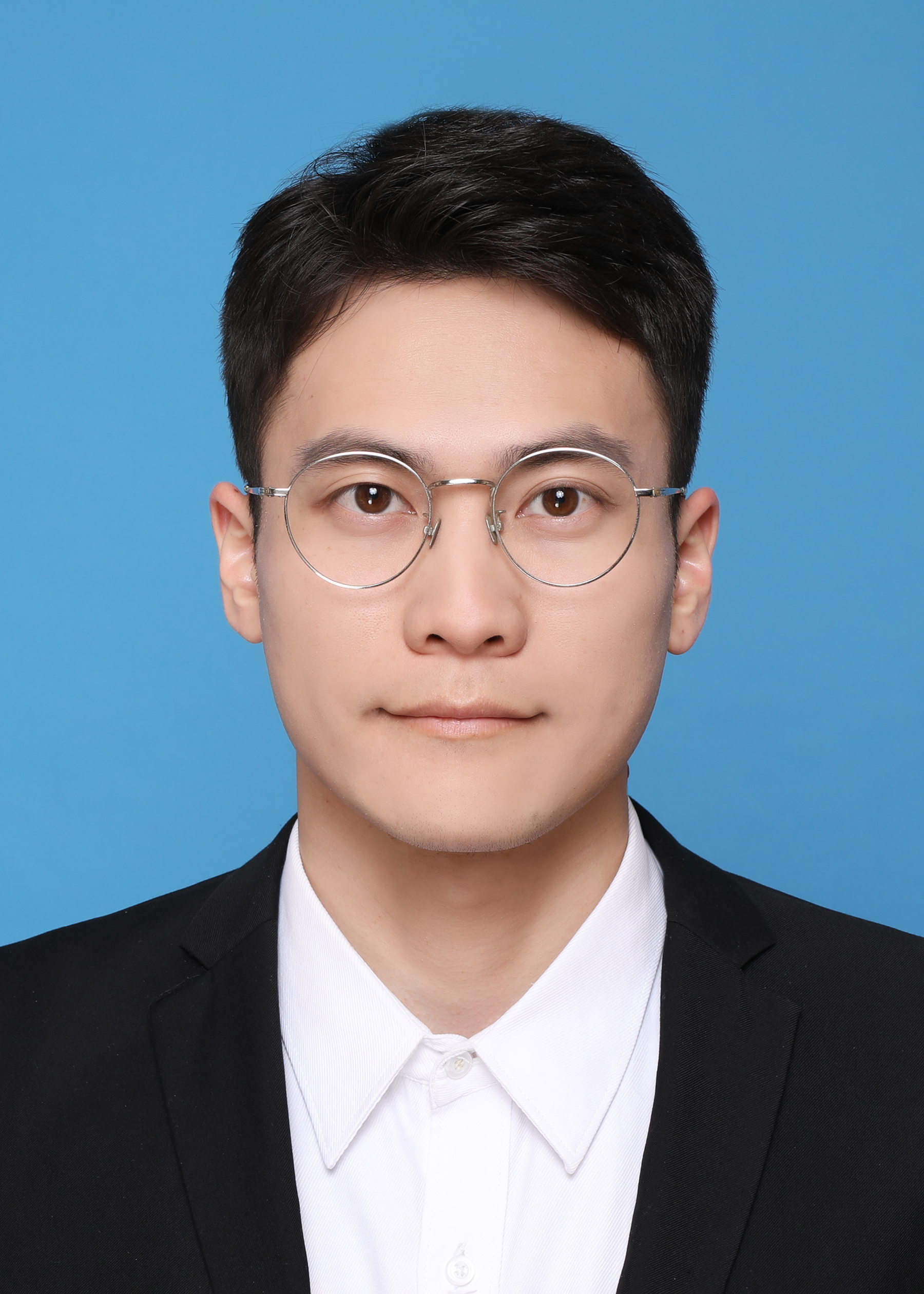}}]{Yuchao Wang} received the B.S. degree in mathematics and applied mathematics from China University of Geosciences, Wuhan, China, in 2018 and the M.S. degree in computational mathematics from the University of Chinese Academy of Sciences, Beijing, China, in 2021. He is currently pursuing the Ph.D. degree in computational mathematics at Fudan University in Shanghai, China. His current research interests include tensor computation and randomized algorithms with their applications.	
\end{IEEEbiography}

\begin{IEEEbiography}[{\includegraphics[width=1in,height=1.25in,clip,keepaspectratio]{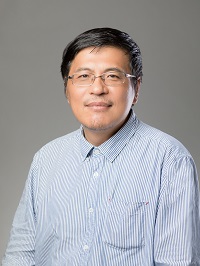}}]{Yimin Wei} received the B.S. degree in computational mathematics from Shanghai Normal University of Shanghai in China, and Ph.D. degree in computational mathematics from Fudan University of Shanghai in China, 1991 and 1997, respectively. 

He was a Lecturer with the Department of Mathematics, Fudan University, from 1997 to 2000, and a Visiting Scholar with the Division of Engineering and Applied Science, Harvard University,	Boston, MA, USA, from 2000 to 2001. From 2001 to 2005, he was an Associate Professor with the School of Mathematical Sciences, Fudan University. He is currently a full Professor with the School of Mathematical Sciences, Fudan University. His current research interests include multilinear algebra and numerical linear algebra with its applications. He is the author of more than 150 technical journal papers and 5 monographs published by Elsevier, Springer, World Scientific and Science Press.
\end{IEEEbiography}

\vfill

\end{document}